\documentclass[lettersize,journal]{IEEEtran}
\usepackage{amsmath,amsfonts}
\usepackage{array}
\usepackage{textcomp}
\usepackage{stfloats}
\usepackage{url}
\usepackage{verbatim}
\usepackage{graphicx}
\usepackage{cite}
\usepackage{subcaption}
\usepackage{multirow}
\usepackage{amsthm}
\usepackage[font={footnotesize}]{subcaption}  
\usepackage{xcolor}
\usepackage{soul}
\usepackage{mathrsfs}
\usepackage{thmtools} 
\declaretheoremstyle[%
spaceabove=0pt,%
spacebelow=4pt,%
headfont=\normalfont\itshape,%
postheadspace=1em,%
qed=\qedsymbol%
]{proofstyle} 
\usepackage{amssymb}
\usepackage{siunitx}

\declaretheorem[name={Proof},style=proofstyle,unnumbered]{prf}

\newtheorem{thm}{Theorem}
\newtheorem{lem}{Lemma}

\newtheorem{asm}{Assumption}

\newtheorem{problem}{Problem}

\usepackage[ruled, vlined, linesnumbered]{algorithm2e}
\makeatletter
\newcommand{\nosemic}{\renewcommand{\@endalgocfline}{\relax}}
\newcommand{\dosemic}{\renewcommand{\@endalgocfline}{\algocf@endline}}
\let\oldnl\nl
\newcommand{\nonl}{\renewcommand{\nl}{\let\nl\oldnl}}
\makeatother
\SetKwBlock{RepeatNoEnd}{repeat}{}

\begin{document}

\title{Adaptive Heterogeneous Compression for Resource-Efficient Federated Knowledge Distillation}
\author{Chenwang Liu, Yijun Liu, Chang Liu, Xu Zhang, and Pengchao Han
\IEEEcompsocitemizethanks{\IEEEcompsocthanksitem Chenwang Liu, Yijun Liu, Chang Liu, and Pengchao Han (corresponding author) are with the School of Information Engineering, Guangdong University of Technology, Guangzhou 510006, China (email: liuchenwang@mails.gdut.edu.cn, 3124002274@mails.gdut.edu.cn, liuchang@gdut.edu.cn, hanpengchao@gdut.edu.cn).\\
Xu Zhang is with the School of Communication and Information Engineering, Chongqing University of Posts and Telecommunications, Chongqing 400065, China (email: zhangxu@cqupt.edu.cn).
}%
}



\maketitle

\begin{abstract}

Federated learning (FL) enables privacy-preserving distributed model training but faces challenges from heterogeneous model architectures and limited communication resources at the network edge. Federated knowledge distillation (FedKD) alleviates model heterogeneity by combining prototype-wise parameter aggregation and knowledge transfer across heterogeneous models. However, transmitting gradients still introduces considerable communication overhead, while existing compression approaches typically apply a uniform strategy across clients and ignore their diverse model characteristics and resource capacities.
To address this issue, we propose a heterogeneous compression framework for FedKD that enables each client to select a compression strategy from a candidate strategy set. We formulate the compression strategy selection problem as a non-stationary stochastic multi-armed bandit (MAB), where each arm corresponds to a compression strategy. An efficiency-aware reward is designed by jointly considering local optimization improvement, global knowledge alignment, and execution time. Based on this formulation, we develop an Adaptive heterogeneouS Compression algorithm for fEderated kNowledge Distillation (ASCEND), which employs an exponential moving average (EMA)-enhanced $\epsilon$-greedy policy to balance exploration and exploitation.
Experimental results on multiple datasets demonstrate that ASCEND effectively adapts to heterogeneous model and resource settings, reducing communication overhead and training time while maintaining competitive model accuracy.

\end{abstract}


\section{Introduction}\label{sec:introduction}

\IEEEPARstart{F}{ederated} learning (FL) has emerged as a promising distributed learning paradigm for privacy-preserving model training, enabling multiple clients to collaboratively learn a shared model~\cite{kairouz2021advances}. 
However, FL inherently suffers from the heterogeneity of client devices and network conditions. In practice, participating clients, such as mobile devices and Internet of Things (IoT) sensors, often exhibit diverse computational resources, memory availability, and communication bandwidths. These disparities make it necessary for clients to adopt heterogeneous model architectures tailored to their local system capabilities, which limits the applicability of traditional parameter-averaging-based FL algorithms that require identical model architectures~\cite{li2020federated}.

Federated knowledge distillation (FedKD) has emerged as a promising approach to address model architecture heterogeneity in FL. In typical FedKD frameworks~\cite{wu2022communication}, a prototype denotes a specific model architecture adopted by a group of clients. The server performs hybrid aggregation that incorporates parameter averaging for clients within the same prototype and logit-based knowledge distillation (KD) across different prototypes. However, repeatedly exchanging full gradients introduces significant communication overhead, especially for large-scale models and resource-constrained edge devices.

Gradient compression provides an effective solution to alleviate communication overhead in FedKD by transmitting only a subset of gradient entries. Representative methods include Top-$K$, Random-$K$, and Periodic-$K$~\cite{shi2019convergence, xu2022detached, stich2018local, lin2017deep}.
However, most existing works adopt a uniform compression strategy across clients and training rounds, neglecting system heterogeneity. First, clients with different computational capabilities require different optimal strategies. For instance, computationally intensive strategies such as Top-$K$ may benefit powerful devices but overload resource-constrained ones, while lightweight methods such as Random-$K$ may degrade performance. Second, model heterogeneity further affects strategy selection, as Top-$K$ incurs higher sorting cost for large models, whereas it may outperform Random-$K$ for small models due to better update quality.
These limitations raise a key question:

\textbf{Key Question}: How can compression strategies be selected for clients with diverse model architectures and resource capacities to achieve an effective balance between training efficiency and learning performance?

To answer this question, we propose a heterogeneous compression framework for FedKD, where each client adaptively selects a compression strategy from a candidate strategy set according to its model characteristics and resource capacities. Furthermore, we formulate the heterogeneous compression strategy selection as an optimization problem that aims to maximize efficiency-aware performance improvement. However, solving this problem is challenging. First, the performance gain of each compression strategy is revealed after execution, making the corresponding reward unavailable in advance. Second, the optimal compression strategy may change throughout training, where performance-oriented strategies may be preferred in early stages while efficiency-oriented strategies may become more desirable in later stages.

To overcome these challenges, we propose an Adaptive heterogeneouS Compression algorithm for fEderated kNowledge Distillation (ASCEND). Specifically, we formulate the client-specific compression strategy selection problem as a non-stationary stochastic multi-armed bandit (MAB) problem, where each arm corresponds to a candidate compression strategy. A reward function is designed to jointly capture model improvement and execution efficiency by considering local training loss, global knowledge alignment, and time cost. An exponential moving average (EMA)-enhanced $\epsilon$-greedy policy is adopted to balance exploration and exploitation during adaptive strategy selection.

The main contributions of this paper are as follows:
\begin{itemize}

\item \textbf{Heterogeneous Compression Framework for FedKD}. 
We propose a heterogeneous compression framework for FedKD that enables each client to independently select a customized compression strategy from a candidate compression set under heterogeneous model architectures and resource conditions.
We provide a convergence analysis under non-convex loss functions, yielding $O(1/\sqrt{T})$ convergence rates.

\item \textbf{Compression Strategy Optimization Problem}. 
We formulate a heterogeneous compression strategy optimization problem that maximizes cumulative efficiency-aware performance, explicitly capturing the tradeoff among compression cost, communication cost, and learning performance.

\item \textbf{ASCEND Algorithm with Theoretical Guarantee}. 
We model compression strategy selection as a non-stationary stochastic MAB problem and propose ASCEND for adaptive strategy selection under dynamic training conditions. Under a piecewise-stationary approximation, we derive a sublinear regret bound of $O(1/\sqrt{MT})$,
providing a theoretical guarantee for the long-term performance of ASCEND.

\item \textbf{Experimental Validation}. 
Extensive experiments on both a real-world platform and  simulation  settings demonstrate that ASCEND adapts effectively to diverse model and resource settings, significantly reducing communication overhead and training time while maintaining high accuracy.
\end{itemize}

The remainder of this paper is organized as follows. 
Section~II reviews related work. 
Section~III presents the system model and preliminaries, including FL, FedKD, and gradient compression with residual accumulation. 
Section~IV introduces the proposed heterogeneous compression framework and provides a convergence analysis. 
Section~V proposes the ASCEND algorithm and its analysis.
Section~VI presents the experimental evaluation of the ASCEND algorithm. 
Finally, Section~VII concludes the paper.

\section{Related Work} \label{sec:related-work}

In this section, we review the literature on FedKD and communication-efficient FL via gradient compression. 
\subsection{Federated Knowledge Distillation}

FedKD is an effective approach to address model heterogeneity in FL by exchanging knowledge representations, such as soft predictions and intermediate features~\cite{wu2022communication,hinton2015distilling,han2025rethinking,11372971}.
A recent survey~\cite{han2025rethinking} revisits KD in collaborative machine learning, highlighting the importance of effective knowledge exchange under heterogeneous learning tasks.

Existing FedKD methods mainly differ in their knowledge sharing mechanisms. FedMD~\cite{li2019fedmd} utilizes a shared public dataset to align prediction outputs among heterogeneous client models, while FedGen~\cite{zhu2021fedgen} introduces a generative model to synthesize proxy samples when public datasets are unavailable. 
Beyond prediction-level knowledge transfer, representation-based FedKD methods exploit intermediate feature information. For example, FedProto~\cite{tan2022fedproto} introduces class-level feature representations extracted from local models and exchanges these semantic representations among clients.
FedTGP~\cite{zhang2023fedtgp} constructs and optimizes global class-level feature representations on the server side.
Instead of directly aggregating local representations, FedTGP learns more discriminative global semantic representations to improve representation aggregation. 
Some studies further consider the personalization issue caused by diverse client characteristics. For example, 
KTpFL~\cite{zhang2021parameter} incorporates KD into personalized FL by transferring knowledge from the global model to client-specific models.

Although these methods effectively mitigate model heterogeneity, the limited amount of shared knowledge may restrict global performance. To improve knowledge aggregation capability, FedDF~\cite{lin2020ensemble} introduces a hybrid aggregation framework that combines parameter averaging among homogeneous models with KD across heterogeneous models. However, FedDF requires clients to upload gradients for ensemble aggregation, incurring substantial communication overhead in bandwidth-constrained federated environments.

\subsection{Communication-Efficient FL via Gradient Compression}
Communication overhead is one of the major challenges in FL, especially when clients need to frequently exchange high-dimensional gradients in bandwidth-constrained networks.
Gradient compression has been widely studied to reduce communication cost by transmitting compact representations of gradients while preserving learning performance~\cite{lin2017deep}. 

Top-$K$ compression is one of the most representative approaches, where only gradients with the largest magnitudes are uploaded. DGC~\cite{lin2017deep} further incorporates momentum correction, local gradient clipping, and residual accumulation to compensate for information loss caused by aggressive compression. These techniques improve communication efficiency while maintaining convergence performance.
To reduce the computational overhead of coordinate selection, Random-$K$ compression randomly selects gradient entries for transmission~\cite{xu2022detached}. Although it avoids the sorting operation required by Top-$K$, it may introduce additional variance during optimization. Periodic-$K$ compression~\cite{stich2018local} addresses this issue by periodically selecting different coordinates, ensuring that all gradients are updated within a predefined number of iterations.

Adaptive compression methods dynamically adjust compression configurations according to training dynamics or system conditions to achieve a better trade-off between communication efficiency and model performance.
AdaComp~\cite{chen2018adacomp} introduces adaptive residual gradient compression, where the compression rate is dynamically adjusted according to gradient residual activity to balance compression efficiency and model accuracy.
For a fixed compression method, Han et al.~\cite{han2020adaptive} propose an online learning method that adaptively adjusts the Top-$K$ compression ratio according to gradient statistics during training, reducing communication overhead while preserving convergence performance.

Although existing compression approaches significantly reduce communication overhead, most methods optimize a single compression mechanism, ignoring heterogeneous model architectures, data characteristics, and system resources.  They do not consider adaptive selection among multiple heterogeneous compression strategies. In FedKD scenarios, different compression strategies may provide different performance-efficiency trade-offs for different clients. Therefore, adaptive selection among multiple compression strategies remains an open problem, motivating our heterogeneous compression framework and online strategy selection approach.

\section{System Model and Preliminaries}
\label{sec:system_models_preliminaries}
In this section, we first present a brief overview of FL and FedKD. We then introduce gradient compression strategies with residual accumulation, which serve as the building blocks of our heterogeneous compression framework.

\subsection{Federated Learning}
An FL system consists of a central server and a client set $\mathcal{N}$ with total $N=|\mathcal{N}|$ clients. Each client $n\in \mathcal{N}$ possesses a local dataset $\mathcal{D}_n$ of size $|\mathcal{D}_n|$. 
Each client trains a local model parameterized by $\boldsymbol{\omega}$ to minimize an empirical risk function  over its private dataset, defined as
\begin{equation}\label{eq:local_loss}
f_n(\boldsymbol{\omega})
=
\frac{1}{|\mathcal{D}_n|}
\sum_{i=1}^{|\mathcal{D}_n|}
\ell_n(\boldsymbol{\omega}; x_i^n),
\end{equation}
where $\ell_n(\boldsymbol{\omega}; x_i^n)$ denotes the sample-wise loss incurred by the model parameters $\boldsymbol{\omega}$ on the $i$th local training sample $x_i^n \in \mathcal{D}_n$. 

In FL, all clients aim to collaboratively learn a global model parameterized by $\boldsymbol{\omega}$ by solving the following weighted optimization problem~\cite{mcmahan2017communication}:
\begin{equation}
\min_{\boldsymbol{\omega}} f(\boldsymbol{\omega}) = \sum_{n=1}^{N} \frac{|\mathcal{D}_n|}{\sum_{n=1}^{N}|\mathcal{D}_n|} f_n(\boldsymbol{\omega}).
\end{equation}

Classical FL algorithms such as FedAvg~\cite{li2019convergence}  perform server-side aggregation by averaging local model updates from participating clients.
However, such approaches typically require identical model architectures across clients, limiting their applicability in scenarios with model architecture heterogeneity.

\subsection{Federated Knowledge Distillation}
\label{sec:Federated Knowledge Distillation}
FedKD is effective in addressing model architecture heterogeneity in FL by enabling knowledge transfer among clients with diverse model architectures~\cite{wu2022communication,hinton2015distilling,li2019fedmd,lin2020ensemble}.
To enable knowledge transfer among heterogeneous models, KD is performed using a small unlabeled public dataset
$\mathcal{D}_{\mathrm{pub}}$. We define a prototype as a specific model architecture. Different clients in FedKD may share the same prototype. Let $\mathcal{N}_p$ denote the set of clients adopting prototype $p$. Classical FedKD methods aggregate local client models through a hybrid strategy that combines parameter averaging among clients with the same model architecture and logit averaging across heterogeneous models~\cite{lin2020ensemble}.

In FedKD, each training round $t$ proceeds as follows. At the start of each training round, each prototype $p$ with parameter $\boldsymbol{\omega}_t^p$ is broadcast to the corresponding clients to process distributed model training as follows.
\begin{itemize}
\item \textbf{Local training}: 
Each client $n$ performs local optimization on its private dataset and computes a stochastic local gradient:
\begin{equation}
\boldsymbol{g}_t^n
=
\nabla f_n(\boldsymbol{\omega}_t^n;\varsigma_t^n),
\end{equation}
where $\boldsymbol{\omega}_t^n$ denotes the local model parameters of client $n$ at communication round $t$, and $\varsigma_t^n$ represents the mini-batch sampled from the local dataset of client $n$.

Instead of uploading the updated local model, each client transmits the computed gradient 
$\boldsymbol{g}_t^n$ to the server. Upon receiving the local gradients, the server updates the corresponding local models by gradient descent and obtain the updated model $\tilde{\boldsymbol{\omega}}_t^n$. 

\item \textbf{Parameter averaging}: 
After obtaining the updated local models 
$\{\tilde{\boldsymbol{\omega}}_t^n\}_{n=1}^{N}$,
the server performs parameter averaging for each model prototype $p$:

\begin{equation}
\bar{\boldsymbol{\omega}}_t^p
=
\sum_{n\in \mathcal{N}_p}
\frac{1}{|\mathcal{N}_p|}
\tilde{\boldsymbol{\omega}}_t^n,
\quad \forall p,
\label{eq:parameter_averaging}
\end{equation}
where  $\bar{\boldsymbol{\omega}}_t^p$ represents the aggregated prototype model at round $t$.

\item \textbf{Global knowledge construction (logit averaging)}:
The server evaluates each model on the public dataset $\mathcal{D}_{\mathrm{pub}}$ to obtain its logits:
\begin{equation}
z^n(x) = \phi^n(\tilde{\boldsymbol{\omega}}_t^n; x), \quad \forall x \in D_{\text{pub}},
\end{equation}
where $\phi^n(\boldsymbol{\omega};x)$ denotes the logit output of client model $n$ with parameters $\boldsymbol{\omega}$ for input sample $x$. 
The global logits are computed by averaging the logits of all clients:
\begin{equation} \label{eq:avg_logit}
\bar{z}(x) = \frac{1}{N} \sum_{n=1}^{N} z^n(x).
\end{equation}
The corresponding global soft targets are obtained by applying the softmax function:
\begin{equation}
\tilde{y}(x) = \mathrm{softmax}\big(\bar{z}(x)\big).
\end{equation}

\item \textbf{Global KD}: To enable each prototype to learn the global knowledge, for each prototype $p$, the server distills the global knowledge $\tilde{y}(x)$ into the aggregated prototype model $\bar{\boldsymbol{\omega}}_t^p$ by minimizing the KL divergence:
\begin{align} \label{eq:kl_loss}
\mathcal{L}_{\text{KD}}^p\!({\bar{\boldsymbol{\omega}}}^p_t)\!=\!\frac{1}{|D_{\text{pub}}|}\!\sum_{x \in\!D_{\text{pub}}}\!\mathcal{L}_{\text{KL}}\!\left(\!\text{softmax}\bigl(\psi^p(\bar{\boldsymbol{\omega}}^p_t; x)\bigr), \tilde{y}(x)\!\right)\!,
\end{align}
where $\psi^p(\boldsymbol{\omega};x)$ denotes the logit output of prototype model $p$ with parameters $\boldsymbol{\omega}$ for input sample $x$.

After distillation, the prototype parameters are updated to $\boldsymbol{\omega}_{t+1}^p $ and transmitted back to the corresponding clients:
\begin{equation}
\boldsymbol{\omega}_{t+1}^n \leftarrow {{\boldsymbol{\omega}}}_{t+1}^p, \quad \forall n \in \mathcal{N}_p.
\end{equation}

Each client uses the received model to start the subsequent training round.
\end{itemize}

\subsection{Gradient Compression with Residual Accumulation}
In FedKD, clients are required to transmit their local model gradients to the server, which incurs substantial communication overhead, particularly for large-scale models. To alleviate this burden, various gradient compression strategies with residual accumulation have been proposed to reduce communication costs while preserving convergence guarantees~\cite{wen2017terngrad,lin2017deep,aji2017sparse}. In this section, we first review representative gradient compression techniques, including Top-$K$, Random-$K$, and Periodic-$K$. Then, we introduce the server-side aggregation of the compressed gradients, followed by the analysis to motivate our heterogeneous compression framework.

\subsubsection{Client-Side Compression with Residual Accumulation}
To compensate for information loss introduced by gradient compression, each client maintains a residual vector $\boldsymbol{e}_{t-1}^n$ at the end of round $t-1$, which accumulates the gradient entries that were not transmitted in previous rounds. At round $t$, the input to the compression operator $\mathcal{C}(\cdot;\cdot)$ is an augmented update vector $\boldsymbol{u}_t^n$, defined as the sum of the current gradient and the residual:
\begin{equation}
\boldsymbol{u}_t^n = \boldsymbol{e}_{t-1}^n + \boldsymbol{g}_t^n.
\label{eq:residual_addition_III}
\end{equation}
For a compression strategy 
$s_t^n \in \{\text{Top-$K$}, \text{Random-$K$},$ $\text{Periodic-$K$}\}$ 
, the compressed gradient on client $n$ is 
\begin{equation}
\hat{\boldsymbol{g}}_t^n
=
\mathcal{C}\!\left(\boldsymbol{u}_t^n; s_t^n\right).
\label{eq:compressed to u}
\end{equation}

The resulting vector $\hat{\boldsymbol{g}}_t^n$ is associated with an index set
$\boldsymbol{I}_t^n \subseteq \{1,\dots,d\}$ of cardinality $K$, where
$d$ denotes the gradient dimension and $\left| \boldsymbol{I}_t^n \right| = K$.
For convenience, we define a reconstruction operator $\mathcal{B}(\cdot;\cdot)$ that maps the compressed representation back to the original gradient dimension. In $\mathcal{B}(\hat{\boldsymbol{g}}_t^n;\boldsymbol{I}_t^n)
\in \mathbb{R}^{d}$,
the entries indexed by $\boldsymbol{I}_t^n$ are filled with the transmitted gradient values and the remaining entries are set to zero.

Typical compression strategies are summarized as follows.
\begin{itemize}
\item \textbf{Top-$K$ Compression}: 
Each client selects the $K$ entries of $\boldsymbol{u}_t^n$ with the largest magnitudes~\cite{lin2017deep,aji2017sparse,haddadpour2021federated}:
\begin{equation}
\boldsymbol{I}_t^n = \mathrm{TopK}(\boldsymbol{u}_t^n, K).
\label{eq:topk_selection_general}
\end{equation}
Top-$K$ compression preserves the most informative gradient directions but incurs a relatively high computational cost due to sorting, with computational complexity of $\mathcal{O}(d \log K)$~\cite{lin2017deep,aji2017sparse}.

\item \textbf{Random-$K$ Compression}: 
Each client randomly samples $K$ indices from all gradient dimensions~\cite{wen2017terngrad,alistarh2017qsgd,basu2019qsparse}:
\begin{equation}
\boldsymbol{I}_t^n = \mathrm{RandomK}(\boldsymbol{u}_t^n, K).
\label{eq:randomk_selection_general}
\end{equation}
This strategy has low computational complexity, i.e., $\mathcal{O}(K)$, but it may introduce high gradient variance and low model performance~\cite{wen2017terngrad,alistarh2017qsgd,haddadpour2021federated}.

\item \textbf{Periodic-$K$ Compression}:
Each client maintains a binary indicator vector $\boldsymbol{b}_t^n \in \{0,1\}^d$ indicating whether a coordinate has been selected within the current
period that contains $T_p$ iterations~\cite{haddadpour2021federated}. The selected index set is
\begin{equation}
\boldsymbol{I}_t^n = \operatorname{PeriodicK}(\boldsymbol{b}_t^n, K),
\label{eq:periodicK}
\end{equation}
where unvisited coordinates are prioritized to ensure periodic coverage of all coordinates, followed by random completion if necessary. After transmission, the selected coordinates are marked as visited, and the time-flag vector $\boldsymbol{b}_t^n$ is updated accordingly. Once all coordinates have been selected at least once, the flag vector is reset to initiate a new periodic cycle.

Periodic-$K$ compression avoids costly sorting and ensures periodic updates of all gradient coordinates, achieving a favorable balance between computation efficiency and convergence stability. Its computational cost is dominated by index selection, with a time complexity of $\mathcal{O}(K)$.
\end{itemize}

After compression, the client transmits the compressed gradients $\hat{\boldsymbol{g}}_t^n$ together with its corresponding index set $\boldsymbol{I}_t^n$ to the server.

To ensure convergence of the federated training process, the residual vector $\boldsymbol{e}_t^n$ is updated by preserving the untransmitted gradient entries:
\begin{equation}
\boldsymbol{e}_t^n
=
\boldsymbol{u}_t^n
-
\mathcal{B}
(\hat{\boldsymbol{g}}_t^n,\boldsymbol{I}_t^n).
\label{eq:residual_update_III}
\end{equation}
This residual accumulation effectively compensates for information loss due to compression and is crucial for achieving stable convergence.

\subsubsection{Server-Side Processing of Compressed Gradients}
Upon receiving the compressed representation $(\hat{\boldsymbol{g}}_t^n,\boldsymbol{I}_t^n)$ from participating clients, the server reconstructs the $d$-dimensional compressed gradient vector using the reconstruction operator:
\begin{equation}
\tilde{\boldsymbol{g}}_t^n
=
\mathcal{B}
(\hat{\boldsymbol{g}}_t^n,\boldsymbol{I}_t^n).
\label{eq:server_tecovery}
\end{equation}

The reconstructed gradients are  used to update the corresponding models for global aggregation and KD:
\begin{equation}
\tilde{\boldsymbol{\omega}}_t^n
=
\boldsymbol{\omega}_t^n
-
\eta_t
\tilde{\boldsymbol{g}}_t^n,
\label{eq:compressed_model_update}
\end{equation}
where $\eta_t$ denotes the learning rate in round $t$

\begin{figure}[t]
    \centering

    \begin{minipage}{\linewidth}
        \centering
        \includegraphics[height=0.8cm]{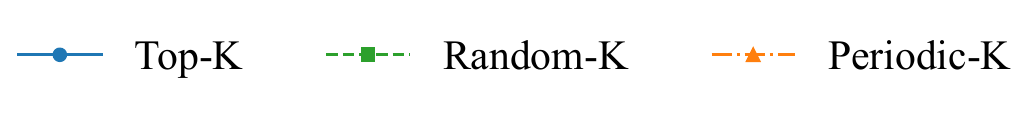}
    \end{minipage}


    \begin{subfigure}[t]{0.32\linewidth}
        \centering
        \includegraphics[width=\linewidth]{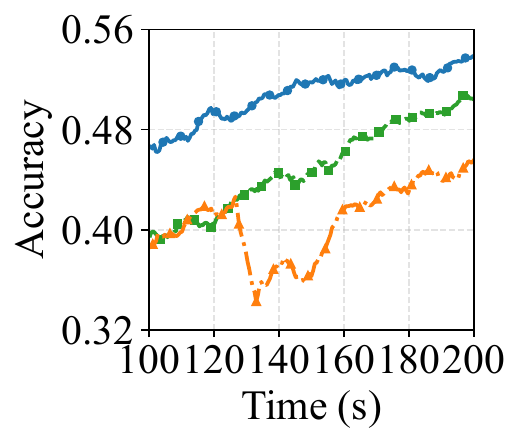}
        \caption{$K = 10^5$}
    \end{subfigure}
    \hfill
    \begin{subfigure}[t]{0.32\linewidth}
        \centering
        \includegraphics[width=\linewidth]{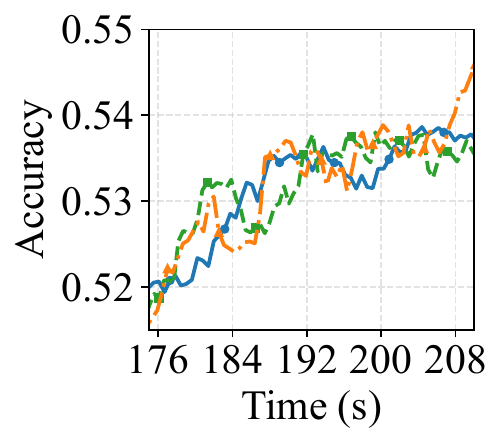}
        \caption{$K=5\times10^5$}
    \end{subfigure}
    \hfill
    \begin{subfigure}[t]{0.32\linewidth}
        \centering
        \includegraphics[width=\linewidth]{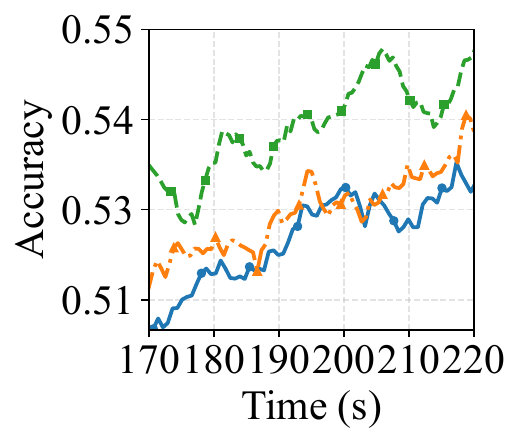}
        \caption{$K=2\times10^6$}
    \end{subfigure}

    \caption{Model performance of three compression strategies under different $K$ (dataset: CIFAR-10, $\alpha = 0.5$).}
    \label{fig:motivation_time_accuracy}
\end{figure}

\subsubsection{Motivation for Heterogeneous Compression}
The three gradient compression strategies, namely Top-$K$, Random-$K$, and Periodic-$K$, exhibit different trade-offs between compression overhead and model performance. 
To investigate whether a uniform compression strategy is suitable for heterogeneous FedKD, we evaluate these strategies under different compression levels and model prototypes on CIFAR-10 and MNIST.

As shown in Fig.~\ref{fig:motivation_time_accuracy}, the preferable compression strategy varies with the compression level. Specifically, Top-$K$ tends to achieve better accuracy under aggressive compression, while Random-$K$ becomes more competitive when a larger communication budget is available due to its lower selection overhead.
Furthermore, Fig.~\ref{fig:motivation_model} demonstrates that the optimal strategy also depends on model architectures under the same compression setting, where Top-$K$ is desirable for a smaller model, e.g., LeNet5Half, and Random-$K$ is better for a larger model, e.g., LeNet5.

These results indicate that no single compression strategy can consistently achieve the best trade-off across heterogeneous clients. Therefore, enforcing a uniform compression strategy may lead to suboptimal performance in FedKD, motivating the proposed heterogeneous compression framework.

\begin{figure}[t]
    \centering

    \begin{subfigure}[t]{0.32\linewidth}
        \centering
        \includegraphics[width=\linewidth]{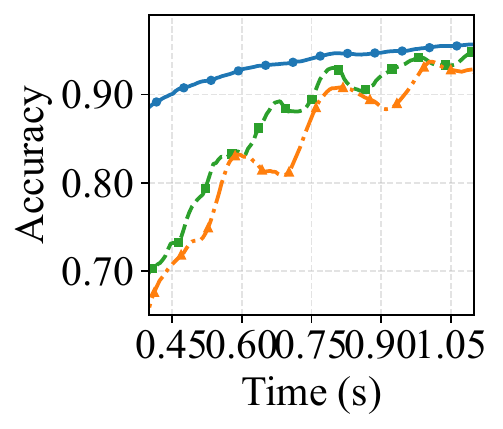}
        \caption{LeNet5Half}
    \end{subfigure}
    \hfill
    \begin{subfigure}[t]{0.32\linewidth}
        \centering
        \includegraphics[width=\linewidth]{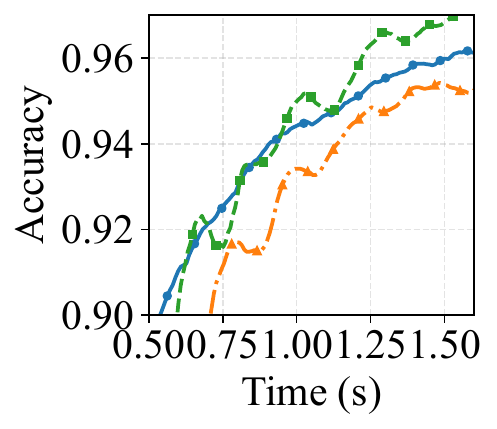}
        \caption{LeNet5}
    \end{subfigure}
    \hfill
    \begin{minipage}[t]{0.24\linewidth}
        \centering
        \includegraphics[width=\linewidth]{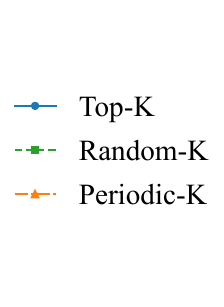}
    \end{minipage}

    \caption{Model performance of three compression strategies under different model prototypes
    (dataset: MNIST, $\alpha = 0.5$, $K = 1 \times 10^{3}$).} 
    \label{fig:motivation_model}
\end{figure}

\section{FedKD with Heterogeneous Compression} 
\label{sec:heterogeneous-compression}
This section introduces a heterogeneous compression framework that allows clients to adopt different compression strategies based on their model architectures and computational resources. We then provide a convergence analysis for the proposed framework and discuss the motivation for adaptive strategy selection.

\subsection{Heterogeneous Compression Framework}
For each client, we define the candidate compression set $\mathcal{S}$ as a collection of three representative strategies:
\begin{equation}
\mathcal{S}
=
\{
\mathrm{Top}\text{-}K,
\mathrm{Random}\text{-}K,
\mathrm{Periodic}\text{-}K
\}.
\label{eq:M}
\end{equation}
Each element in $\mathcal{S}$ corresponds to a distinct compression strategy, characterized by different trade-offs among model accuracy, computational overhead, and communication efficiency.

Figure~\ref{have_compression_new} illustrates the system architecture of the proposed FedKD framework with heterogeneous gradient compression. As shown, each client independently performs local training on its private dataset, constructs residual-accumulated gradients, and applies a compression strategy from $\mathcal S$ before transmitting compressed gradients and corresponding indices to the server. The server first reconstructs the compressed gradients and updates the corresponding client models. Then, the updated client models are utilized for prototype-wise parameter aggregation among clients with the same model architecture and global logit averaging across all client models. Finally, the server performs KD on the public dataset $\mathcal{D}_{\mathrm{pub}}$ to align global knowledge and broadcasts the updated prototype models back to the corresponding clients~\cite{wu2022communication,hinton2015distilling}.

\begin{figure}[t]
    \centering
    \includegraphics[width=\linewidth]{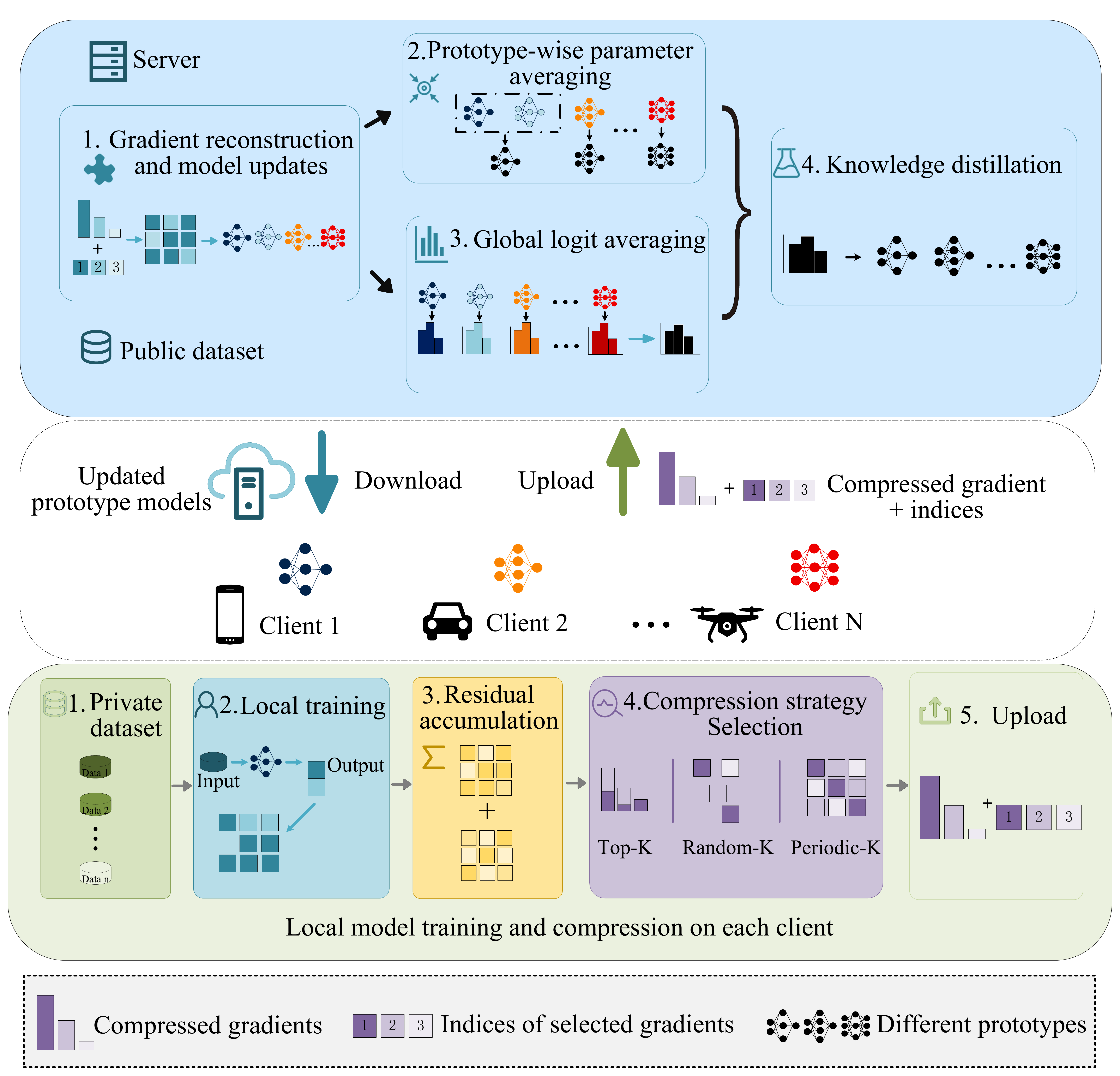}
    \caption{The system architecture of the heterogeneous compression framework for FedKD. }
    \label{have_compression_new}
\end{figure}

\subsection{Convergence Analysis of FedKD with Heterogeneous Compression}
This section analyzes the convergence behavior of the proposed heterogeneous compression framework in FedKD.
Unlike conventional FL methods with homogeneous models, FedKD maintains multiple model prototypes to handle model architecture heterogeneity.
Therefore, we conduct the convergence analysis from the perspective of each model prototype. 

The following assumptions are introduced to analyze the convergence of the proposed heterogeneous compression framework.

\begin{asm}[Smoothness for Local Training Objective]
\label{asm:Smoothness-l}
For each client $n$, the objective function
$f_n(\boldsymbol{\omega})$ 
is $L_l$-smooth~\cite{stich2018local,koloskova2020unified}. Specifically, there exists a constant $L_l>0$ such that
\begin{equation}
\left\|
\nabla f_n(\boldsymbol{\omega})
-
\nabla f_n(\boldsymbol{\omega}')
\right\|
\leq
L_l
\left\|
\boldsymbol{\omega}
-
\boldsymbol{\omega}'
\right\|,
\quad
\forall
\boldsymbol{\omega},\boldsymbol{\omega}' .
\label{eq:smoothnesslocal training}
\end{equation}
\end{asm}
\begin{asm}[Smoothness for KD Objective]
\label{asm:Smoothness-p}
For each prototype $p$, the objective function
$\mathcal{L}_{\text{KD}}^p(\boldsymbol{\omega})$
is $L_p$-smooth \cite{10606337}. Specifically, there exists a constant $L_p>0$ such that
\begin{equation}
\left\|
\nabla \mathcal{L}_{\text{KD}}^p(\boldsymbol{\omega})
-
\nabla \mathcal{L}_{\text{KD}}^p(\boldsymbol{\omega}')
\right\|
\leq
L_p
\left\|
\boldsymbol{\omega}
-
\boldsymbol{\omega}'
\right\|,
\quad
\forall
\boldsymbol{\omega},\boldsymbol{\omega}' .
\label{eq:smoothnessp}
\end{equation}
\end{asm}

\begin{asm}[Unbiased Stochastic Gradient]
\label{asm:Unbiased Stochastic Gradient}
For each client $n\in\mathcal{N}_p$  and round $t$, the stochastic gradient
$\boldsymbol{g}_t^n$ satisfies~\cite{stich2018local} 
\begin{equation}
\mathbb{E}
\left[
\boldsymbol{g}_t^n
\mid
\boldsymbol{\omega}_t^p
\right]
=
\nabla f_n(\boldsymbol{\omega}_t^p).
\label{eq:unbiased_gradient}
\end{equation}
\end{asm}

\begin{asm}[Bounded Stochastic Gradient Variance]
\label{asm:Bounded Stochastic Gradient Variance}
For any client $n$ and round $t$, the variance of stochastic gradients is bounded by a constant
$\sigma^2>0$~\cite{stich2018local,woodworth2020minibatch}, i.e.,
\begin{equation}
\mathbb{E}
\left[
\left\|
\boldsymbol{g}_t^n
-
\nabla f_n(\boldsymbol{\omega}_t^p)
\right\|^2
\right]
\leq
\sigma^2.
\label{eq:gradient_variance}
\end{equation}
\label{ASM:Bounded Stochastic Gradient Variance}
\end{asm}
\begin{asm}[Bounded Error-Feedback Residual]
\label{asm:Bounded Error-Feedback Residual}
For each client $n$ and round $t$, the accumulated residual satisfies~\cite{karimireddy2019errorfeedback,richtarik2021ef21}
\begin{equation}
\mathbb{E}
\left[
\|
\boldsymbol e_t^n
\|^2
\right]
\leq
\Gamma_{\max}
,
\label{eq:error_feedback_bound}
\end{equation}
where $\Gamma_{\max}$ denotes the maximum residual amplification factor among all candidate compression strategies.
\end{asm}
Assumption~\ref{asm:Bounded Error-Feedback Residual} ensures that these untransmitted errors remain stable and do not diverge after accumulating across multiple rounds of residual propagation. 

For theoretical analysis, we define the joint optimization objective of prototype $p$ as
\begin{equation}
F_p(\boldsymbol{\omega}^p_t)
=
f_p(\boldsymbol{\omega}^p_t)
+
\lambda
\mathcal{L}_{KD}^{p}
(\boldsymbol{\omega}^p_t),
\end{equation} 
where $\lambda$ balances the local optimization objective and the global knowledge alignment objective.
Thus, we have the following theorem for the convergence of FedKD with heterogeneous compression for non-convex loss functions.
\begin{thm}[Convergence of FedKD with Heterogeneous Compression]
\label{thm:nonconvex_convergence}
Under Assumptions~\ref{asm:Smoothness-l}--~\ref{asm:Bounded Error-Feedback Residual}, for each prototype $p$, after $T$ communication rounds, the proposed FedKD framework with heterogeneous compression satisfies
\begin{equation}
\frac1T
\sum_{t=0}^{T-1}
\mathbb {E}
\left[
\|
\nabla F_p(\boldsymbol{\omega}_t^p)
\|^2
\right]
\leq
\frac{
4
(
F_p(\boldsymbol{\omega}_0^p)-F_p^\ast+C_k
)
}
{3\eta_tT}
+
\frac{4L}{3}
\eta_t C_g  ,
\label{eq:final_nonconvex_bound}
\end{equation}

By setting
$\eta_t=O(1/\sqrt{T})$,
the convergence rate becomes
\begin{equation}
\frac{1}{T}
\sum_{t=0}^{T-1}
\mathbb{E}
\left[
\left\|
\nabla F_p(\boldsymbol{\omega}_t^p)
\right\|^2
\right]
=
O
\left(
\frac{1}{\sqrt T}
\right).
\label{eq:final_tate}
\end{equation}

\end{thm}

The proof of the theorem can be found in our online report {\color{red}[*]}.
Theorem~\ref{thm:nonconvex_convergence}
shows that the proposed method converges  with a convergence rate of
$O(1/\sqrt{T})$, which equals the convergence rate of existing FL methods.

\subsection{Motivation for Adaptive Compression Algorithm}
Although the heterogeneous compression framework enables clients to adopt different compression strategies according to their characteristics, determining the appropriate strategy for each client remains challenging.
Specifically, the performance of each compression strategy cannot be known before execution and may vary during training due to dynamic model states and system conditions.

Therefore, an online strategy selection mechanism is required to allow each client to adaptively identify suitable compression strategies based on observed training feedback. This motivates the formulation of compression strategy selection as a non-stationary stochastic MAB problem, where each compression strategy corresponds to an arm and the observed efficiency-aware performance gain serves as the reward.

\section{FedKD with Adaptive Heterogeneous Compression}
\label{sec:adaptive-selection}

In this section, we first formulate the heterogeneous compression strategy optimization problem in FedKD. Then, we formulate the problem as a MAB and propose the adaptive heterogeneous compression algorithm. Finally, we analyze the regret of the online algorithm.

\subsection{Heterogeneous Compression Strategy Optimization Problem}
The heterogeneous compression strategy optimization problem aims to maximize efficiency-aware model performance by determining the compression strategy 
$s^n_t \in \mathcal{S}$ for all clients $n \in \mathcal{N}$ at round $t$.
To achieve this, we define the model performance gain of client $n$ at round $t$ as $\mathcal{G}(s_t^n)$. In FedKD, the time overhead covers the local training time, the computation time for gradient compression, and the communication time for gradient transmission. Thus, the problem is defined as

\begin{problem}[Heterogeneous Compression Strategy Optimization Problem]
\label{pro:hetero_com}
In FedKD with heterogeneous compression, all clients aim to
\begin{equation}
\label{eq:adaptive_obj}
\max_{\{s^n_t\}_{n\in\mathcal{N}}}
\sum_{n=1}^{N} \sum_{t=1}^{T}  \mathcal{G}(s^n_t).
\end{equation}
\end{problem}
Solving Problem~\ref{pro:hetero_com} is challenging. First, the performance gain $\mathcal{G}(s^n_t)$ is stochastic and unknown in advance, which is only revealed after the compression strategy $s^n_t$ is selected and executed by client $n$. 
Moreover, the reward distribution may vary over time due to non-stationary factors such as data heterogeneity, changes in loss values, and dynamic communication conditions. 
This observation motivates us to address Problem~\ref{pro:hetero_com} in an online decision-making framework.
Specifically, we formulate adaptive compression strategy selection as a non-stationary stochastic MAB problem, where each client selects a compression strategy from a finite candidate set at each communication round and observes the corresponding reward as feedback. 
The objective is to learn an effective strategy-selection policy that maximizes the cumulative expected performance gain over time.

\subsection{MAB Formulation}
We first formulate the time overhead and reward. Then, we reformulate Problem~\ref{pro:hetero_com} as an MAB. In this reformulation, each client $n$ decides a sequence of compression strategies $\{s_t^n\}$ over training rounds.

\subsubsection {Time overhead}
We formulate the client-side training time for each training round $t$ as
\begin{equation}
\label{eq:adaptive_time}
\tau_t^n(s_t^n) = \nu^n({s_t^n},K) + \zeta^n(\boldsymbol{\omega}_t^n) + \vartheta^n(K)
,
\end{equation}
where $\nu^n(\cdot,\cdot)$ denotes the compression time that depends on the selected strategy, model architecture, and compression level $K$, $\zeta^n(\cdot)$ represents the local training time of client $n$,  and $\vartheta(\cdot)$ represents the transmission time of the compressed gradients and their indices.

\subsubsection{Reward}
The reward of client $n$ is defined as an efficiency-aware performance gain:
\begin{equation}
\label{eq:adaptive_teward}
\mathcal{G}(s_t^n)
=
\frac{\Delta \mathcal{L}_n(t)}{\tau_t^n(s_t^n)},
\end{equation}
where $\Delta \mathcal{L}_n(t)$ quantifies the model improvement induced by the selected compression strategy. 
Specifically, for each client $n$ with prototype $p$ at round $t$, $\Delta \mathcal{L}_n(t)$ is measured by jointly considering global knowledge alignment that reflects global knowledge improvement and local optimization progress that captures the client-specific loss reduction on private data:
\begin{equation}
\label{eq:deltaL}
\begin{aligned}
\Delta \mathcal{L}_n(t)
= {} & \rho \Big(
\mathcal{L}_{\text{KD}}^p({\bar{\boldsymbol{\omega}}}^p_{t-1}) -
\mathcal{L}_{\text{KD}}^p({\bar{\boldsymbol{\omega}}}^p_t)
\Big) \\
& + \beta \Big(
f_n(\boldsymbol{\omega}_{t-1}^n)
-
f_n(\boldsymbol{\omega}_t^n)
\Big),
\end{aligned}
\end{equation}
where $\rho,\beta>0$ denote the weighting coefficients.
The global KD loss for prototype $p$, i.e., $\mathcal{L}_{\text{KD}}^p({\bar{\boldsymbol{\omega}}}^p_t)$ and the local loss of client $n$, i.e., $f_n(\boldsymbol{\omega}_t^n)$ are defined in \eqref{eq:kl_loss} and \eqref{eq:local_loss}, respectively. The KD loss component is computed at the server and then is fed back to the corresponding clients.

By normalizing the combined performance gain with the execution time
$\tau_t^n$, the proposed reward encourages compression strategies that achieve faster convergence with lower computational and communication overhead.
Furthermore, we apply a clipping function to the reward to avoid
numerical instability caused by large reward fluctuations:
\begin{equation}
\mathcal{G}(s_t^n) \leftarrow \operatorname{clip}\!\left(\mathcal{G}(s_t^n), \mathcal{G}_{\min}, \mathcal{G}_{\max}\right). 
\end{equation}
In our experiments, the thresholds are set to
$\mathcal{G}_{\min}=-1$ and
$\mathcal{G}_{\max}=1$.

\subsubsection{MAB Formulation} The reward of each compression strategy is time-varying because it evolves with dynamic model performance during training. Therefore, we model the proposed strategy selection as a non-stationary stochastic MAB. Unlike adversarial settings where rewards may change arbitrarily, the reward transition in our framework follows the intrinsic optimization dynamics. We define the MAB problem as follows.
\begin{problem}[Non-stationary Stochastic MAB formulation]\label{pro:mab}
    We reformulate Problem~\ref{pro:hetero_com} as a non-stationary stochastic MAB problem.
\begin{itemize}
\item \textbf{Arms:} each arm corresponds to a compression strategy $s_t^n \in \mathcal{S}$;

\item \textbf{Policy:} In each training round $t$, client $n$ selects a compression strategy $s_t^n$ to balance exploration and exploitation:
\begin{equation}
s_t^n=
\begin{cases}
\text{Random}(\mathcal{S}),
& \text{with probability }\epsilon_t,\\
\arg\max\limits_{s\in\mathcal{S}} Q_t^n(s),
& \text{with probability }1-\epsilon_t.
\end{cases}
\label{eq:epsilon_greedy}
\end{equation}
where $\epsilon_t\in(0,1)$ denotes the exploration rate and $Q_t^n(s)$ denotes the local utility estimate of strategy $s$ maintained by client $n$ at  round $t$.

\item \textbf{Utility update:} Each client maintains a local utility estimate $Q_t^n(s')$ for each strategy $s'\in\mathcal{S}$. To adapt to non-stationary stochastic reward variation, we employ an EMA-based utility estimator within an $\epsilon$-greedy policy \cite{sutton2018reinforcement}. After observing the reward
$\mathcal{G}(s_t^n)$,
client $n$ updates its utility estimate as follows:
\begin{equation} \label{eq:q_update}
Q_{t+1}^n(s') = 
\begin{cases}
\big(1-\eta_q\big)Q_t^n(s') + \eta_q \mathcal{G}\left(s_t^n\right) & \text{if } s_t^n = s', \\
Q_t^n(s') & \text{if } s_t^n \neq s'.
\end{cases}
\end{equation}
where $\eta_q \in (0,1]$ denotes the EMA step size.
\end{itemize}
\end{problem}

\subsection{Algorithm Design}
In this section, we propose the ASCEND algorithm for Problem~\ref{pro:mab}, where each client independently optimizes its compression strategy in an online manner to cope with dynamic model performance and computational efficiency trade-offs in heterogeneous FedKD.
Specifically, the proposed algorithm combines an adaptive exploration policy and a rollback-based stability safeguard.

\subsubsection{Adaptive Exploration Policy}
We employ a two-phase exploration policy to adapt to different training stages:
\begin{itemize}
\item {Warm-up Phase:}
To avoid early-stage bias in utility estimation, each client performs a fixed
number of warm-up trials for every compression strategy in $\mathcal{S}$.
Let $S = |\mathcal{S}|$ denote the number of available compression strategies.
During the warm-up stage, each strategy $s^n_t \in \mathcal{S}$ is selected exactly
$h$ times by each client, resulting in a total of $hS$ warm-up rounds.
The selection order is randomly permuted to prevent systematic ordering bias.

Each strategy $s^n_t$ is treated as an independent arm and is associated with a reward
estimate $\mathcal{G}(s_t^n)$. 
The reward of strategy $s^n_t$ is estimated as the empirical mean:
\begin{equation}
\label{eq:Q0}
Q_{0}^n(s')
=
\frac{1}{h}
\sum_{t=1}^{h |\mathcal{S}|}
\mathcal{G}(s_t^n) \mathbb{I}_{s_t^n=s'}.
\end{equation}
Here, $\mathbb{I}_{s_t^n=s'}$ indicates whether the selected strategy $s_t^n$ corresponds to the target strategy $s'$. The empirical mean operation provides an initial estimate of the average reward for each strategy for the subsequent adaptive EMA-enhanced $\epsilon$-greedy updates.
As a result, at the end of the warm-up stage, each client maintains $S$ initialized utility
values $\{Q_0^n(s)\}_{s \in \mathcal{S}}$, one for each compression strategy.

\item {Adaptive Phase:} 
For rounds $t > hS$, the compression strategy is selected following \eqref{eq:epsilon_greedy}.
Here, $Q_t^n(s_)$ is initialized by $Q_0^n(s)$ and updated according to ~\eqref{eq:q_update}.
We define the exploration rate in \eqref{eq:epsilon_greedy} as
\begin{equation}
\label{eq:maxr}
\epsilon_t = \min\left(1,\,\frac{c}{\sqrt{t}}\right),
\end{equation}
where $c$ is a fixed hyperparameter.
\end{itemize}

\subsubsection{Rollback-Based Stability Safeguard}
To prevent unstable convergence caused by overly aggressive or rapidly fluctuating
compression strategy selections, we incorporate a rollback-based stability safeguard
inspired by~\cite{yin2018byzantine, karimireddy2020scaffold}. 

A stability check is performed by comparing the global loss between two consecutive
rounds. For any prototype $p$, if a sudden loss surge is detected, i.e.,
\begin{equation}
\mathcal{L}_{\text{KD}}^p({\bar{\boldsymbol{\omega}}}^p_t)>\gamma \mathcal{L}_{\text{KD}}^p({\bar{\boldsymbol{\omega}}}^p_{t-1}),
\label{eq:rollback_trigger}
\end{equation}
where $\gamma > 1$ is a predefined stability threshold,
the system activates the following stabilization procedure:
\begin{itemize}
    \item \textbf{Model Rollback:} All client models and prototype models are restored
    to their parameters from the previous round $t-1$. 
    \item \textbf{Strategy Unification:}
    To provide a conservative fallback under unstable conditions, all clients adopt the Top-$K$ strategy for the next round.
\end{itemize}

By reverting the system to a previously stable state and temporarily constraining strategy heterogeneity, this rollback mechanism effectively mitigates performance degradation caused by erratic compression decisions and enhances the robustness of the proposed adaptive compression framework.

\begin{algorithm}[t]
\caption{ASCEND algorithm}
\label{alg:adaptive_fkd}
\footnotesize
\KwIn{
Client set $\mathcal{N}$, public dataset $\mathcal{D}_{\mathrm{pub}}$, warm-up length $h$, exploration rate $\epsilon_t$, learning rate $\eta_t$, $\eta_q$.
}
\KwOut{Clients' models $\{\boldsymbol{\omega}_T^{n}\}_{n \in \mathcal{N}}$.}

Initialize prototype models $\boldsymbol{\omega}_0^{p}, \forall p$ and distribute them to corresponding clients;\\
Initialize residual vectors $\boldsymbol{e}_0^{n}\leftarrow\boldsymbol{0},\ \forall n \in \mathcal{N}$ and utility estimates $Q_0^n(s)\leftarrow0,\ \forall n \in \mathcal{N},\forall s\in\mathcal{S}$;

\For{$t=1$ \KwTo $T$}{


    \nonl\textbf{Client-side Local Training and Gradient Compression:}\\
    \For{each client $n \in \mathcal{N}$ \textbf{in parallel}}{

        Compute local gradient
        $\boldsymbol{g}_t^{n}$\; \label{algline:local_training}

        Accumulate residuals:
        $\boldsymbol{u}_t^{n}\leftarrow\boldsymbol{e}_{t-1}^{n}+\boldsymbol{g}_t^{n}$\;

        Select compression strategy $s_t^n$ by \eqref{eq:strateg_selection};\label{algline:select_strategy}

        Compress gradients:
        $(\hat{\boldsymbol{g}}_t^{n}, \boldsymbol{I}_t^{n})
        \leftarrow \mathcal{C}(\boldsymbol{u}_t^{n}; s_t^n)$\; \label{algline: compress_grad}

        Update residuals:
        $\boldsymbol{e}_{t}^{n}
        \leftarrow
        \boldsymbol{u}_{t}^{n}
        -
        \mathcal{B}(\hat{\boldsymbol{g}}_t^n;\boldsymbol{I}_t^n)$\;

        Upload $(\hat{\boldsymbol{g}}_t^{n}, \boldsymbol{I}_t^{n})$ to the server\; \label{algline:upload}
    }

    \nonl\textbf{Server-side Prototype-wise Model Aggregation:}\\
    Reconstruct compressed gradients:
    $\tilde{\boldsymbol{g}}_t^{n}
    \leftarrow
    \mathcal{B}(\hat{\boldsymbol{g}}_t^n,\boldsymbol{I}_t^n),
    \forall n \in \mathcal{N}$; \label{algline:resconstruct}
    
    \For{each prototype $p$}{

        Update local models:
    $\tilde{\boldsymbol{\omega}}_t^{n}
        \leftarrow
        \boldsymbol{\omega}_t^{n}
        -
        \eta_t
        \tilde{\boldsymbol{g}}_t^{n},
        \forall n\in\mathcal{N}_{p}$;\label{algline:Update local models}

        Aggregate prototype model by \eqref{eq:parameter_averaging};
        \label{algline:Aggregate prototype model}
    }

    \nonl\textbf{Server-side KD:}\\
    Evaluate all models on $\mathcal{D}_{\mathrm{pub}}$ and compute averaged logits $\bar{z}(x)$ by \eqref{eq:avg_logit};\label{algline:logit_avg}
    
    Update each prototype model by minimizing $\mathcal{L}_{\text{KD}}^p({\bar{\boldsymbol{\omega}}}^p_t)$\;\label{algline:kd}

    Broadcast updated prototype models and global loss to clients\;\label{algline:broadcast}

    \If{$\mathcal{L}_{\text{KD}}^p({\bar{\boldsymbol{\omega}}}^p_t)>\gamma \mathcal{L}_{\text{KD}}^p({\bar{\boldsymbol{\omega}}}^p_{t-1})$}{
        Roll back to round $t\!-\!1$ and enforce unified compression\label{algline:roll_back}; 
        
       }

    \nonl\textbf{Client-side Reward and Utility Update:}\\
    \For{each client $n \in \mathcal{N}$ \textbf{in parallel}}{
        Compute reward $\mathcal{G}(s_t^n)$ by (\ref{eq:adaptive_teward})\;\label{algline:compute_teward}

        Update utility estimate:
       $Q_{t+1}^n(s_t^n)
        =
        (1-\eta_q)\,Q_t^n(s_t^n)
        +
        \eta_q\,\mathcal{G}_t^n(s_t^n)$\label{algline:update_utility}.
    }
}
\end{algorithm}

\subsubsection{Algorithm Description}
The ASCEND algorithm is summarized in Algorithm~\ref{alg:adaptive_fkd}.
In each round $t$, participating clients first perform local training (line~\ref{algline:local_training}).
Each client then selects a compression strategy according to the EMA-enhanced $\epsilon$-greedy policy:
\begin{align} \label{eq:strateg_selection}
        s_t^n\!=\!\begin{cases}
        \text{the next in shuffled sequence},  t \le h|\mathcal{S}|, \\[0.5ex]
        \text{Random}(\mathcal{S}),  t > h|\mathcal{S}| \text{ with probability } \epsilon_t, \\[0.5ex]
        \displaystyle \arg\max_{s\in\mathcal{S}} Q_t^n(s),  t > h|\mathcal{S}| \text{ with probability } 1-\epsilon_t;
        \end{cases}
\end{align}
(line~\ref{algline:select_strategy}), compresses its residual-accumulated gradients (line~\ref{algline: compress_grad}), and uploads the resulting compressed gradient representation to the server (line~\ref{algline:upload}).

To aggregate heterogeneous compressed gradients, the server first reconstructs the received compressed gradients (line~\ref{algline:resconstruct}) and updates the corresponding client models (line~\ref{algline:Update local models}). The updated models within each prototype $p$ are then aggregated through parameter averaging to obtain the prototype models (line~\ref{algline:Aggregate prototype model}). Afterwards, the server adopts a logit-based KD mechanism using a small public dataset $\mathcal{D}_{\mathrm{pub}}$. Specifically, the server evaluates all models on $\mathcal{D}_{\mathrm{pub}}$ and averages their logits to construct global soft targets, which capture shared knowledge across different prototypes (line~\ref{algline:logit_avg}). Each prototype model further aligns its output distribution with the global soft targets via KL-divergence-based KD (line~\ref{algline:kd}). 
The distilled prototype models and the global loss are subsequently broadcast to their corresponding clients for the next round of local training (line~\ref{algline:broadcast}).
In addition, clients compute the instantaneous reward  (line~\ref{algline:compute_teward}) and update the utilities (line~\ref{algline:update_utility}) to facilitate the subsequent strategy selection and global optimization. 

\subsection{Algorithm Analysis}
In this section, we analyze the regret of the proposed ASCEND algorithm. 

We first define the regret used to evaluate ASCEND. Since the reward is affected by stochastic gradients, data sampling, and model evolution, it is treated as a random variable. For notational simplicity, we hereinafter use $\mathcal{G}_t^n(s)$ to denote the reward obtained by client $n$ when selecting compression strategy $s$ in communication round $t$. The expected reward is defined as
\begin{equation}
\mu_t^n(s)=
\mathbb{E}
\left[
\mathcal{G}_t^n(s)
\right].
\label{eq:expected_reward}
\end{equation}
Unlike in stationary bandit problems, $\mu_t^n(s)$ may vary throughout training because the optimization states and KD outcomes continuously evolve. Accordingly, the optimal compression strategy for client $n$ in round $t$ is defined as
\begin{equation}
s_t^{n,*}=
\arg\max_{s\in\mathcal S}
\mu_t^n(s).
\label{eq:optimal_strategy}
\end{equation}
The regret of ASCEND for client $n$ over $T$ communication rounds is then defined as
\begin{equation}
\mathcal{R}_T^n=
\sum_{t=1}^{T}
\left[
\mu_t^n(s_t^{n,*})-
\mu_t^n(s_t^n)
\right],
\label{eq:regret}
\end{equation}
which measures the cumulative expected-reward loss relative to the round-wise optimal compression strategies.

The following assumptions are introduced for analyzing the adaptive strategy
selection behavior of ASCEND.
\begin{asm}[Piecewise Stationary Reward]
\label{asm:piecewise_stationary}
The reward process of each compression strategy is assumed to be piecewise stationary during the training process~\cite{garivier2011switching,liu2018change}. Specifically, the entire communication process is divided into $M$ stationary segments:
\begin{equation}
\{\mathcal{T}_1,\mathcal{T}_2,\ldots,\mathcal{T}_M\},
\end{equation}
where the reward distribution remains approximately unchanged within each segment. That is, for client $n$, compression strategy $s\in\mathcal{S}$, and communication round $t\in\mathcal{T}_i$, the expected reward satisfies
\begin{equation}
\mu_t^n(s)
=
\mu_i^n(s), \forall i,
\label{eq:piecewise_stationary}
\end{equation}
where $\mu_i^n(s)$ denotes the stationary expected reward of strategy
$s$ within segment $\mathcal{T}_i$.
\end{asm}

\begin{asm}[Reward Concentration]
\label{asm:reward_concentration}
The stochastic reward obtained from each compression strategy satisfies the sub-Gaussian concentration property~\cite{auer2002finite}. Specifically, for any $n\in\mathcal{N}, s\in\mathcal{S}$, and $x>0$,
\begin{equation}
\Pr
\left(
\left|
\mathcal{G}_t^n(s)
-
\mu_t^n(s)
\right|
>
x
\right)
\leq
2
\exp
\left(
-\frac{x^2}{2\sigma_g^2}
\right), \forall t> h|\mathcal{S}|,
\label{eq:reward_concentration}
\end{equation}
where $\sigma_g^2$ represents the variance of the reward distribution.
\end{asm}

\begin{asm}[Sufficient Strategy Sampling]
\label{asm:sufficient_sampling}
For each client $n$, stationary segment $\mathcal T_i$, and strategy $s\in\mathcal S$, let $r_i$ denote the first communication round of segment $\mathcal T_i$. We define
\begin{equation}
H_{i,t}^n(s)
=
\sum_{x=r_i}^{t}
\mathbb I\{s_x^n=s\}, \forall s\in \mathcal{S}
\end{equation}
as the number of times that strategy $s$ has been selected from round $r_i$ to round $t$ during the segment $\mathcal T_i$.
Then, there exists a constant $\pi_{\min}\in(0,1]$ such that
\begin{equation}
H_{i,t}^n(s)
\geq
\pi_{\min}(t-r_i+1), \forall s\in \mathcal{S}.
\end{equation}
\end{asm}

Based on the above assumptions, we introduce the following lemmas to characterize the estimation errors and regret incurred by the EMA, exploration, and exploitation phases of ASCEND.

We first characterize how the EMA tracks the expected reward within each stationary segment.
\begin{lem}[EMA Error]
\label{lem:ema_tracking}
For any client $n$ and compression strategy $s\in\mathcal S$, within a stationary segment $\mathcal T_i$, the EMA-based utility estimate satisfies
\begin{equation}
\left|
\mathbb E[Q_t^n(s)]
-
\mu_i^n(s)
\right|
\leq
(1-\eta_q)^m
\left|
\mathbb E[Q_{r_i}^n(s)]
-
\mu_i^n(s)
\right|,
\label{eq:ema_tracking_error}
\end{equation}
where 
$m$ represents the number of updates of strategy $s$ during this segment.
\end{lem}
Lemma~\ref{lem:ema_tracking} shows that, within a stationary segment, the expected EMA tracking error decreases geometrically with the number of updates $m$. Thus, strategies that are sampled more frequently can adapt more rapidly to the reward distribution of the current segment.

We then derive separate regret bounds for the exploration and exploitation phases.
\begin{lem}[Regret Upper Bound for Exploration ]
\label{lem:exploration_regret}
Under Assumptions~\ref{asm:piecewise_stationary}--\ref{asm:reward_concentration}, the cumulative regret incurred by exploration over $T$ communication rounds satisfies
\begin{equation}
\mathcal{R}_{\mathrm{explore}}^n(T)
\leq
2c\Delta_{\max}\sqrt T ,
\label{eq:exploration_regret_bound}
\end{equation}
where
\begin{equation}
\Delta_{\max}
=
\mathcal G_{\max}
-
\mathcal G_{\min}
\label{eq:Gmax,min}
\end{equation}
is the maximum instantaneous regret under the clipped reward range, and $c$ is the constant in the exploration rate defined in \eqref{eq:maxr}. Consequently,
\begin{equation}
\mathcal{R}_{\mathrm{explore}}^n(T)
=
O(\sqrt T).
\end{equation}
\end{lem}
Lemma~\ref{lem:exploration_regret} indicates that the cumulative cost of exploration grows sublinearly with the number of communication rounds. Hence, the average regret caused by exploration vanishes as $T$ increases.

\begin{lem}[Regret Upper Bound for Exploitation]
\label{lem:exploitation_regret}
Under Assumptions~\ref{asm:piecewise_stationary}--\ref{asm:sufficient_sampling}, the cumulative regret incurred by exploitation for client $n$ satisfies
\begin{equation}
\mathcal{R}_{\mathrm{exploit}}^n(T)
\leq
C_e
\sum_{i=1}^{M}
\sqrt{|\mathcal T_i|},
\label{eq:exploit_segment_bound}
\end{equation}
where $C_e$ is a constant determined by the EMA update coefficient, reward variation, initialization error, and reward concentration parameter.
Consequently,
\begin{equation}
\mathcal{R}_{\mathrm{exploit}}^n(T)
\leq
C_e\sqrt{MT}.
\label{eq:exploit_final_bound}
\end{equation}
\end{lem}

Lemmas~\ref{lem:exploration_regret} and~\ref{lem:exploitation_regret} jointly characterize the regret incurred by the two adaptive selection phases of ASCEND. 
Under a piecewise-stationary reward model, we derive an $O(\sqrt{T})$ exploration-regret bound and an $O(\sqrt{MT})$ exploitation-regret bound.

Based on the preceding lemmas, we establish the overall regret bound of ASCEND.
\begin{thm}[Regret Upper Bound of ASCEND]
\label{thm:regret}
For each client $n$, the cumulative regret of ASCEND over $T$ communication rounds satisfies
\begin{equation}
\mathcal{R}_{\mathrm{ASCEND}}^n(T)
\leq
h|\mathcal S|\Delta_{\max}
+
2c\Delta_{\max}\sqrt T 
+
C_e\sqrt{MT},
\end{equation}
where the three terms correspond to the regret incurred during the warm-up, exploration, and exploitation phases, respectively. Since the warm-up length $h|\mathcal{S}|$ is fixed, the overall regret satisfies
\begin{equation}
\mathcal{R}_{\mathrm{ASCEND}}^n(T)
=
O(\sqrt{MT}).
\end{equation}
\end{thm}

The proofs of Lemmas~\ref{lem:ema_tracking}--\ref{lem:exploitation_regret} and Theorem~\ref{thm:regret} are provided in {\color{red}[*]}.

Theorem~\ref{thm:regret} demonstrates that ASCEND asymptotically approaches the performance of the round-wise optimal compression strategy in a piecewise-stationary training environment. The regret bound further reveals the impact of environmental non-stationarity: a larger number of stationary segments $M$ leads to a higher tracking cost because ASCEND must repeatedly adapt its utility estimates to changes in the reward distributions.

Combined with Theorem~\ref{thm:nonconvex_convergence}, which establishes convergence under any compression strategy in $\mathcal{S}$, Theorem~\ref{thm:regret} shows that ASCEND can adaptively select communication-efficient compression strategies without compromising the convergence guarantee of the underlying federated optimization process.

\section{Experiments}\label{sec:evaluation}
In this section, we evaluate the effectiveness and robustness of the proposed ASCEND algorithm via experiments.
We introduce the experimental settings, present empirical results on a real-world platform and various experimental evaluations, and conduct sensitivity analysis to investigate the impact of critical hyperparameters.

\subsection{Setup}
We introduce the setup, including federated settings, datasets, models, time overheads, and baselines.

\subsubsection{Federated setting}
We build a real-world platform with 10 Raspberry Pi devices, which act as resource-constrained edge clients controlled by a server.
For experimental evaluation, we evaluate a total of 10 clients, each of which is assigned a local model architecture from the candidate prototype set.
The batch size is set to 32, and the learning rate is 0.02. 
The weighting parameters $\rho$ and $\beta$ in \eqref{eq:deltaL}
are set to 0.6 and 0.4, respectively, and the fixed hyperparameter $c$ in \eqref{eq:maxr} is set to 3 in all experiments.
These default values are chosen to achieve favorable empirical performance, and we conduct sensitivity analysis to investigate the impact of hyperparameter variations.

\subsubsection{Datasets} We conduct experiments on two datasets, namely MNIST~\cite{lecun2002gradient} and CIFAR-10~\cite{krizhevsky2009learning}. 
To simulate realistic non-independent and identically distributed (non-IID) settings,  datasets are partitioned using a Dirichlet distribution with parameter $\alpha$ that controls the degree of data heterogeneity. A smaller $\alpha$ leads to a more skewed class distribution, while a larger $\alpha$ yields a more balanced data distribution.
Unless otherwise specified, we set $\alpha = 0.5$.

\subsubsection{Models}

For the MNIST dataset, we employ two lightweight convolutional neural networks, namely LeNet5 and LeNet5Half, where LeNet5Half is a shallower
and more compact variant of LeNet5.
The total number of trainable parameters is 61,706 for LeNet5 and 15,738 for LeNet5Half, respectively.
For the CIFAR-10 dataset, we consider two architectures, ResNet18 and WResNet40-2.
ResNet18 contains 11,183,562 parameters, while WResNet40-2 contains 2,248,954 parameters. Compared with ResNet18, WResNet40-2 adopts a wider residual architecture with increased width and a different depth-width trade-off, providing a useful complementary model pair for evaluating the proposed adaptive compression framework under heterogeneous model sizes and architectural characteristics.

\subsubsection{Time overheads}
For the platform, we collect the real computation and communication time and feed them back to the ASCEND algorithm for compression strategy determination.
For the experimental evaluation, the computation time, including local training and gradient compression, is instantiated using the real measured execution time of the system on an NVIDIA RTX 3090 GPU and an AMD Ryzen 7 4800H CPU.
All computation times are obtained by averaging over multiple runs. The local training time is reported in Table~\ref{tab:local_training}, while representative compression times for different values of $K$ are summarized in Tables~\ref{tab:mnist_time} and~\ref{tab:cifar_time}. 

\begin{table}[t]
\centering
\caption{Local Training Time per Round (seconds)}
\label{tab:local_training}
\begin{tabular}{c|c|c}
\hline
Dataset & Model & Time (s) \\
\hline
MNIST & LeNet5Half & 0.0025 \\
MNIST & LeNet5 & 0.0045 \\
\hline
CIFAR-10 & WResNet40-2 & 0.120 \\
CIFAR-10 & ResNet18 & 0.194 \\

\hline
\end{tabular}
\end{table}

\begin{table}[t]
\centering
\small
\caption{Gradient Compression Time on MNIST (seconds)}
\label{tab:mnist_time}
\begin{tabular}{c|c|ccc}
\hline
Model & Method & $K=5\times10^2$ & $10^3$ & $5\times10^3$ \\
\hline
\multirow{3}{*}{LeNet5Half}
 & Top-$K$ & 1.70e-7 & 3.40e-7 & 1.70e-6 \\
 & Random-$K$ & 3.50e-8 & 7.00e-8 & 3.50e-7 \\
 & Periodic-$K$ & 1.00e-7 & 2.00e-7 & 1.00e-6 \\
\hline
\multirow{3}{*}{LeNet5}
 & Top-$K$ & 3.05e-7 & 6.10e-7 & 3.05e-6 \\
 & Random-$K$ & 6.00e-8 & 1.20e-7 & 6.00e-7 \\
 & Periodic-$K$ & 1.75e-7 & 3.50e-7 & 1.75e-6 \\
\hline
\end{tabular}
\end{table}

\begin{table}[t]
\centering
\small
\caption{Gradient Compression Time on CIFAR-10 (seconds)}
\label{tab:cifar_time}
\begin{tabular}{c|c|ccc}
\hline
Model & Method & $K=10^5$ & $5\times10^5$ & $2\times10^6$ \\
\hline
\multirow{3}{*}{WResNet40-2}
 & Top-$K$ & 1.3e-4 & 3.0e-4 & 9.0e-4 \\
 & Random-$K$ & 2.0e-5 & 3.0e-5 & 7.5e-5 \\
 & Periodic-$K$ & 6.5e-5 & 1.5e-4 & 5.0e-4 \\
\hline
\multirow{3}{*}{ResNet18}
 & Top-$K$ & 6.5e-4 & 1.5e-3 & 4.5e-3 \\
 & Random-$K$ & 4.0e-5 & 5.0e-5 & 2.0e-4 \\
 & Periodic-$K$ & 3.2e-4 & 7.5e-4 & 2.55e-3 \\
\hline
\end{tabular}
\end{table}

The communication time for client $n$ is modeled as
\begin{equation}
\varphi^n(K) = \frac{K \cdot b_v + K \cdot b_i}{B\times 10^6},
\end{equation}
where $K$ denotes the number of retained gradient coordinates after compression, $b_v$ and $b_i$ represent the numbers of bits used to encode each gradient value and its corresponding index, respectively. 
$B$ denotes the bandwidth.
In our implementation, gradient values are represented using 32-bit floating point precision ($b_v = 32$), while indices are encoded using $\log_2 d$ bits, where $d$ denotes the dimensionality of the original gradient vector.
The default communication bandwidth $B$ is set to 50 Mbps. 

\subsubsection{Baselines}
Since existing adaptive gradient compression mechanisms mainly target compression-ratio adaptation in FL, they are not directly designed for heterogeneous strategy selection in FedKD. Therefore, we compare ASCEND with uniform compression baselines, where all clients adopt the same strategy, and an EXP3-based adaptive compression method.

\subsection{Experimental Results and Analysis} 
We verify the effectiveness of the proposed ASCEND algorithm from multiple perspectives, including its performance on a real-world platform, comparison with different baselines at different levels of compression, robustness under varying levels of data heterogeneity and communication bandwidth, as well as the evolution of client-side strategy selection behaviors during training, and sensitivity analysis to investigate the impact of key hyperparameters.

\subsubsection{Performance on Real-World Platform}
\label{Performance on Real-World Platforms}

To demonstrate the practicality of the proposed ASCEND algorithm, we evaluate its performance on a real-world platform.
The experiments are conducted on the MNIST dataset under different compression levels $K \in \{5\times10^3, 10^4\}$, covering relatively aggressive and moderate compression settings.

\begin{figure}[t]
    \centering
    \begin{subfigure}[t]{0.32\linewidth}
        \centering
         \includegraphics[width=\linewidth]{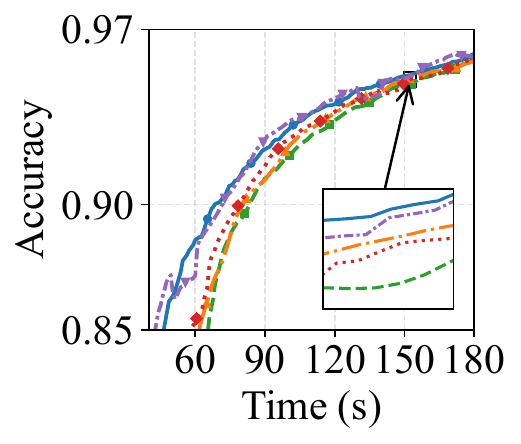}
         \caption{$K=5\times10^3$}
         \label{fig:real,k5000_TIME_PREDICTION}
    \end{subfigure}
    \hfill
    \begin{subfigure}[t]{0.32\linewidth}
        \centering
         \includegraphics[width=\linewidth]{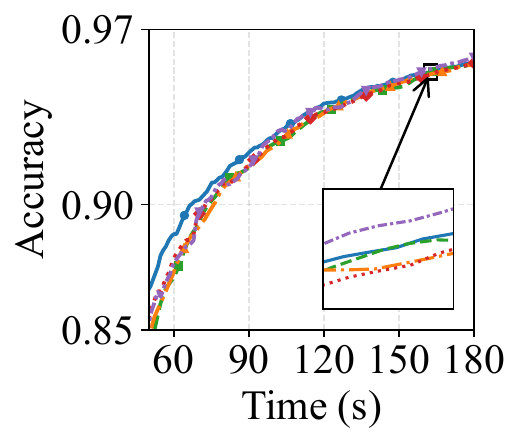}
         \caption{$K=10^4$}
         \label{fig:real,k10000_TIME_PREDICTION}
    \end{subfigure}
    \hfill
    \begin{minipage}[t]{0.24\linewidth}
        \centering
        \includegraphics[width=\linewidth]{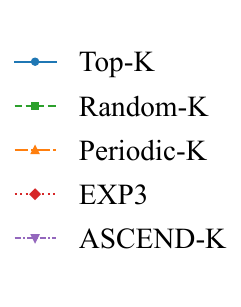}
    \end{minipage}

     \caption{Comparison of different methods on the MNIST dataset deployed on a real-world platform under different compression levels ($\alpha = 0.5$).}
     \label{fig:real_platform_mnist}
\end{figure}

The comparison of different methods under different compression levels on the MNIST dataset is presented in Fig.~\ref{fig:real_platform_mnist}.
We can observe that ASCEND achieves competitive performance across different compression levels, demonstrating its ability to effectively balance communication efficiency and model accuracy.
In particular, ASCEND achieves the best performance at the compression level $K=10^4$.
Furthermore, at the compression level $K=5\times10^3$, ASCEND exhibits faster convergence in the early training stage. This behavior can be attributed to its adaptive strategy selection mechanism, which dynamically selects suitable compression strategies based on the observed utility estimates, thereby improving communication efficiency and accelerating empirical risk reduction.

\subsubsection{Performance Comparison with Baselines}
\label{Performance Comparison of Compression Strategies}

Fig.~\ref{fig:mnist_time_comparison} illustrates the performance on MNIST under different compression levels $K \in \{5\times10^{2},\,10^{3},\,5\times10^{3}\}$, corresponding to compression ratios of approximately $0.81\%$, $1.62\%$, and $8.10\%$ for LeNet5 as the reference model. Fig.~\ref{fig:cifar_time_comparison} presents the results on CIFAR-10 with $K \in \{ 10^5, 5 \times 10^5, 2 \times 10^6\}$, corresponding to retained-gradient ratios of approximately $0.89\%$, $4.47\%$, and $17.88\%$ for ResNet18 as the
reference model.
Generally, increasing $K$ improves convergence speed and final accuracy because more gradient information is preserved. However, the optimal compression strategy varies across different compression levels. For example, on MNIST, Top-$K$ performs better when $K=5\times10^{2}$, while Random-$K$ becomes more competitive at $K=5\times10^{3}$; similarly, on CIFAR-10, Top-$K$ is more effective at $K=10^5$, whereas Random-$K$ shows advantages at larger budgets such as $K=2\times10^6$.
Despite these variations, ASCEND consistently achieves strong performance across all compression levels, demonstrating its capability to adaptively balance the trade-off between communication efficiency and model accuracy.

\begin{figure}[t]
    \centering

    \begin{minipage}{\linewidth}
        \centering
        \includegraphics[height=0.8cm]{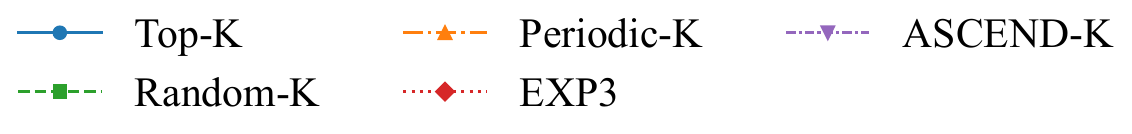}
    \end{minipage}

    \begin{subfigure}[t]{0.32\linewidth}
        \centering
        \includegraphics[width=\linewidth]{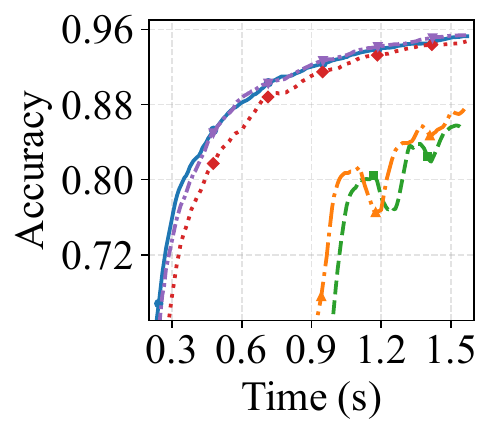}
        \caption{$K=5\times10^2$}
    \end{subfigure}
    \hfill
    \begin{subfigure}[t]{0.32\linewidth}
        \centering
        \includegraphics[width=\linewidth]{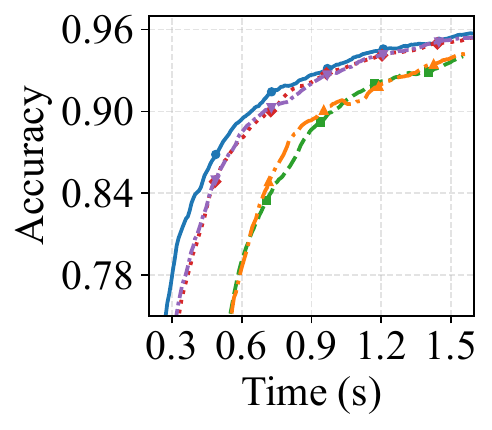}
        \caption{$K=10^3$}
    \end{subfigure}
    \hfill
    \begin{subfigure}[t]{0.32\linewidth}
        \centering
        \includegraphics[width=\linewidth]{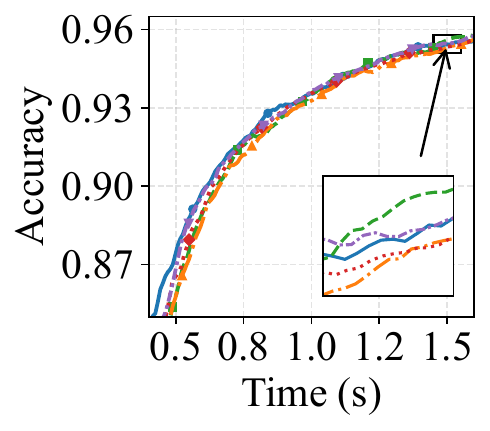}
        \caption{$K=5\times10^3$}
    \end{subfigure}

    \caption{Comparison of different methods on the MNIST dataset under different compression levels ($\alpha = 0.5$).}
    \label{fig:mnist_time_comparison}
\end{figure}

\begin{figure}[t]
    \centering

    \begin{minipage}{\linewidth}
        \centering
        \includegraphics[height=0.8cm]{fig/legend_mnist.pdf}
    \end{minipage}

    \begin{subfigure}[t]{0.32\linewidth}
        \centering
        \includegraphics[width=\linewidth]{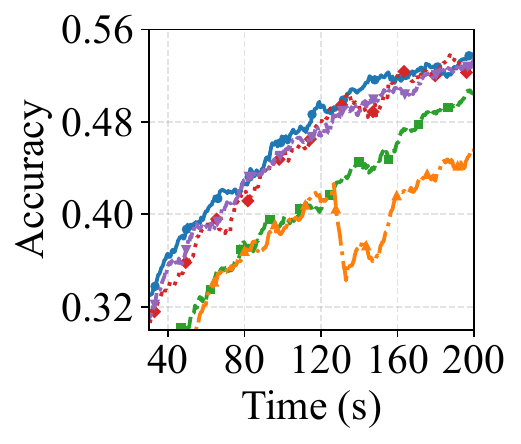}
        \caption{$K = 10^5$}
    \end{subfigure}
    \hfill
    \begin{subfigure}[t]{0.32\linewidth}
        \centering
        \includegraphics[width=\linewidth]{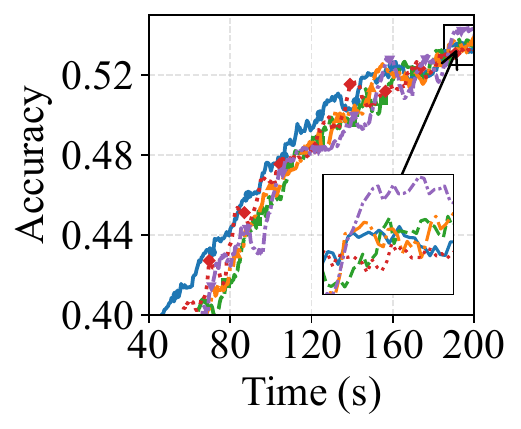}
        \caption{$K=5\times10^5$}
    \end{subfigure}
    \hfill
    \begin{subfigure}[t]{0.32\linewidth}
        \centering
        \includegraphics[width=\linewidth]{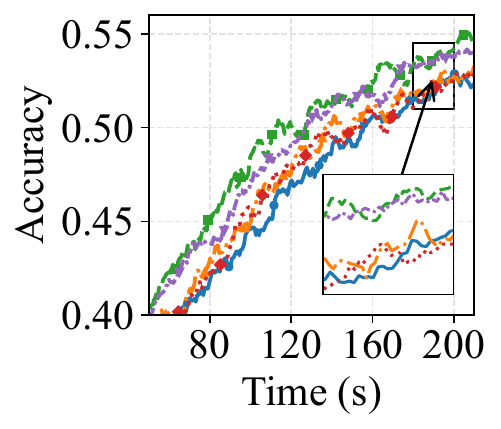}
        \caption{$K=2\times10^6$}
    \end{subfigure}

    \caption{Comparison of different methods on the
    CIFAR-10 dataset under different compression levels ($\alpha = 0.5$).}
    \label{fig:cifar_time_comparison}
\end{figure}

\subsubsection{Performance under Varying Data Heterogeneity}
\label{Performance under Varying Data Heterogeneities}

To further evaluate the robustness of the proposed method under varying degrees of data heterogeneity, we conduct additional experiments on the MNIST dataset, where $\alpha \in \{0.1, 0.5, 1\}$ correspond to highly heterogeneous and nearly IID data distributions, respectively.
In this experiment, the compression level is fixed at $K = 10^{3}$, and the communication bandwidth is set to $50$~Mbps.

As shown in Figs. \ref{fig:mnist_time_comparison} and~\ref{fig:analysis_A},  higher data heterogeneity (e.g., $\alpha = 0.1$) leads to slower convergence and lower final accuracy.
Nevertheless, across all settings, ASCEND consistently achieves competitive or superior performance compared to the best fixed compression strategy.

\begin{figure}[t]
    \centering

    \begin{subfigure}[t]{0.32\linewidth}
        \centering
        \includegraphics[width=\linewidth]{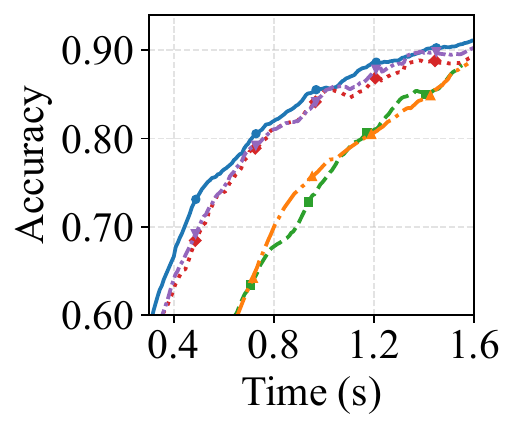}
        \caption{$\alpha = 0.1$}
    \end{subfigure}
    \hfill
    \begin{subfigure}[t]{0.32\linewidth}
        \centering
        \includegraphics[width=\linewidth]{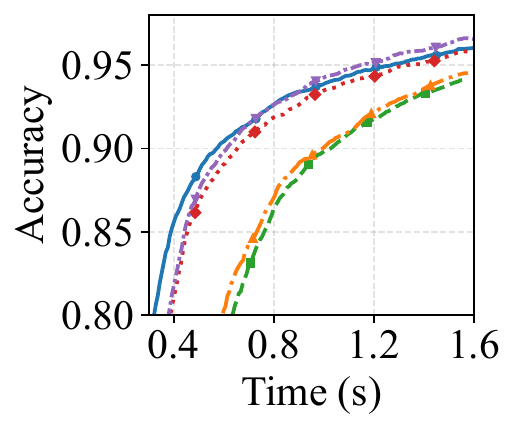}
        \caption{$\alpha = 1$}
    \end{subfigure}
    \hfill
    \begin{minipage}[t]{0.24\linewidth}
        \centering
        \includegraphics[width=\linewidth]{fig/legend_vertical.pdf}
    \end{minipage}

    \caption{Comparison of different methods on the MNIST dataset under different data heterogeneity levels
    ($K = 10^{3}$).}
    \label{fig:analysis_A}
\end{figure}

\subsubsection{Performance under Varying Communication Bandwidths}
\label{Performance under Varying Communication Bandwidths}

To evaluate the robustness of the proposed method under varying communication bandwidths, we conduct additional experiments on the MNIST dataset, where the bandwidth is varied over \{5~Mbps, 50~Mbps, 500~Mbps\}, corresponding to different network conditions.
Here, we set $K = 10^{3}$ and $\alpha = 0.5$.

As shown in Figs.\ref{fig:mnist_time_comparison} and~\ref{fig:Analysis_of_the_bandwidth}, we can observe that a small bandwidth (e.g., 5 Mbps) introduces significant communication delays, under which Top-$K$ tends to perform better. In contrast, higher bandwidth (e.g., 500 Mbps) improves convergence efficiency, making Random-$K$ more competitive.
Nevertheless, across all settings, ASCEND consistently achieves competitive or superior performance, demonstrating its robustness under varying network conditions.

\begin{figure}[t]
    \centering
    \begin{subfigure}[t]{0.32\linewidth}
        \centering
        \includegraphics[width=\linewidth]{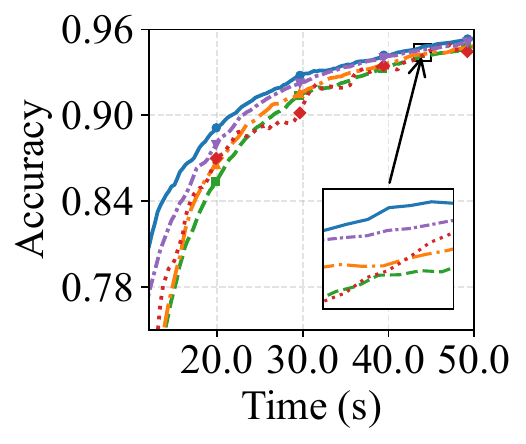}
        \caption{5 \text{Mbps}}
    \end{subfigure}
    \hfill
    \begin{subfigure}[t]{0.32\linewidth}
        \centering
        \includegraphics[width=\linewidth]{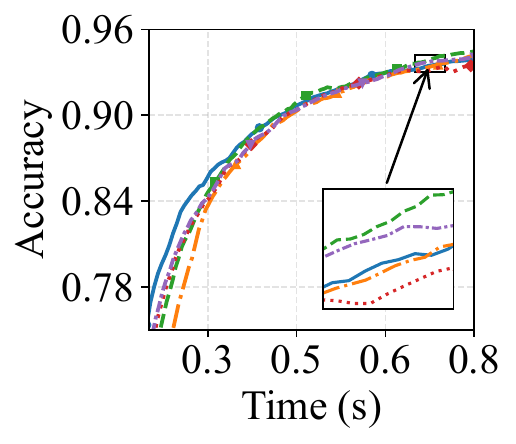}
        \caption{$5\times10^2$ \text{Mbps}}
    \end{subfigure}
    \hfill
    \begin{minipage}[t]{0.24\linewidth}
        \centering
        \includegraphics[width=\linewidth]{fig/legend_vertical.pdf}
    \end{minipage}

    \caption{Comparison of different methods on the MNIST dataset under different communication bandwidths ($\alpha = 0.5$, $K =10^{3}$
    ). }
    \label{fig:Analysis_of_the_bandwidth}
\end{figure}

\subsubsection{Analysis of Client Compression-Strategy Selection}
\label{subsec:client_behavior}

In this section, we analyze the compression strategy selection results of clients  to better understand the selection dynamics of the proposed ASCEND algorithm during the training process.
For the three candidate strategies, i.e., Top-$K$, Random-$K$, and Periodic-$K$, we record the average compression strategy selection counts for different prototypes during training. 

For MNIST with a small $K = 10^{3}$, as illustrated in Fig.~\ref{fig:analysis_behaviors2}, clients equipped with the lightweight LeNet5Half model predominantly select the Top-$K$ strategy, whereas those using the standard LeNet5 model exhibit a more balanced preference between Top-$K$ and Random-$K$. 
This discrepancy can be attributed to differences in computational efficiency across model sizes.
For smaller models such as LeNet5Half, Top-$K$ is more effective since the gradient  dimension is relatively low, making the Top-$K$ selection operation computationally efficient while preserving the most important gradient entries.
In contrast, for larger models such as LeNet5, the cost of identifying Top-$K$ entries becomes higher, and Random-$K$ provides a more efficient alternative with lower selection overhead while maintaining competitive training performance under the same communication constraints. 

In contrast, for a large $K = 5 \times 10^{3}$ on MNIST, as shown in
Fig.~\ref{fig:Analysis of the behaviors3}, both LeNet5Half and LeNet5 clients gradually shift their preference toward the Random-$K$ strategy. At this compression level, more gradient entries can be transmitted, which narrows the performance difference among compression strategies. As a result, the lower selection overhead of Random-$K$ becomes more advantageous, leading to higher computational efficiency during
training.

The compression strategy selection results for CIFAR-10 under smaller and larger $K$ values, i.e., $K=10^5$ and $K=2\times10^6$, are shown in Figs.~\ref{fig:Analysis of the behaviors4} and \ref{fig:Analysis of the behaviors5}, respectively. We can observe the same trend as MNIST.
Overall, these observations indicate that the proposed ASCEND algorithm is capable of autonomously identifying effective compression strategies and adapting its decisions according to model architecture, resource constraints, and compression strategy features.

\begin{figure}[t]
    \centering
    \begin{minipage}{\linewidth}
        \centering
        \includegraphics[height=0.8cm]{fig/legend.pdf}
    \end{minipage}
    \begin{subfigure}[t]{0.48\linewidth}
        \centering
        \includegraphics[width=\linewidth]{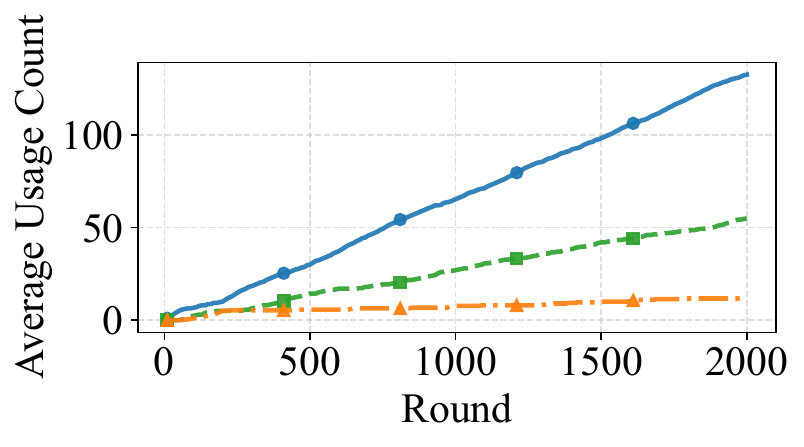}
        \caption{Model prototype = LeNet5Half}
        \label{fig:image1}
    \end{subfigure}
    \hfill
    \begin{subfigure}[t]{0.48\linewidth}
        \centering
        \includegraphics[width=\linewidth]{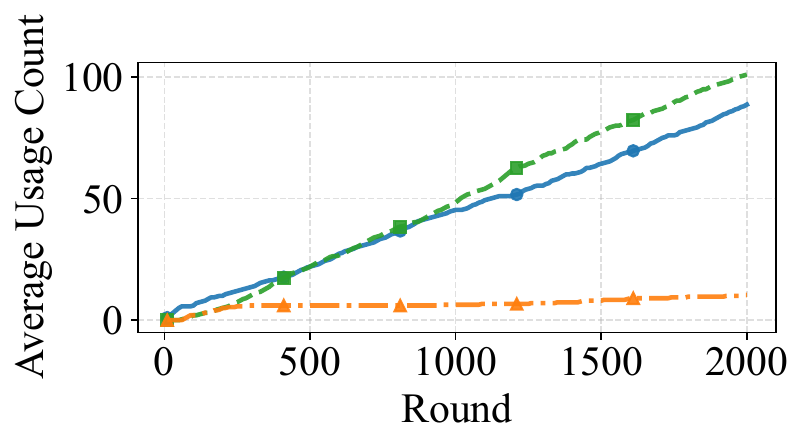}
        \caption{Model prototype = LeNet5}
        \label{fig:image2}
    \end{subfigure}

    \caption{Compression strategy selection results of different model prototypes on MNIST dataset ($\alpha = 0.5$, $K=10^3$).}
    \label{fig:analysis_behaviors2}
\end{figure}

\begin{figure}[t]
    \centering
    \begin{minipage}{\linewidth}
        \centering
        \includegraphics[height=0.8cm]{fig/legend.pdf}
    \end{minipage}
    \begin{subfigure}[t]{0.48\linewidth}
        \centering
        \includegraphics[width=\linewidth]{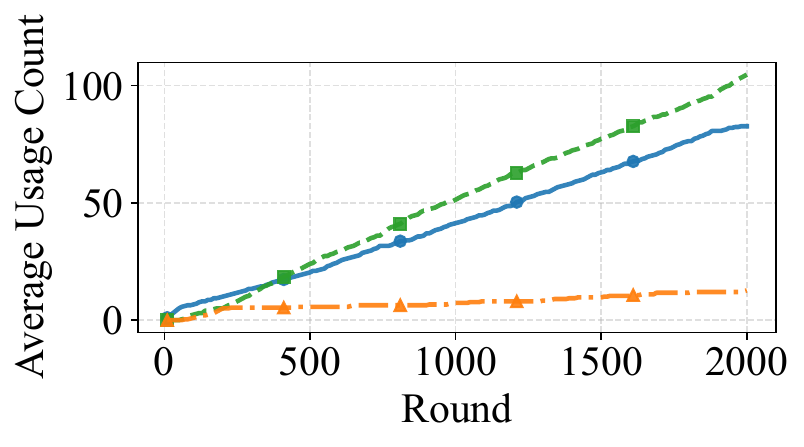}
        \caption{Model prototype = LeNet5Half}
        \label{fig:mnist_024_5000_use}
    \end{subfigure}
    \hfill
    \begin{subfigure}[t]{0.48\linewidth}
        \centering
        \includegraphics[width=\linewidth]{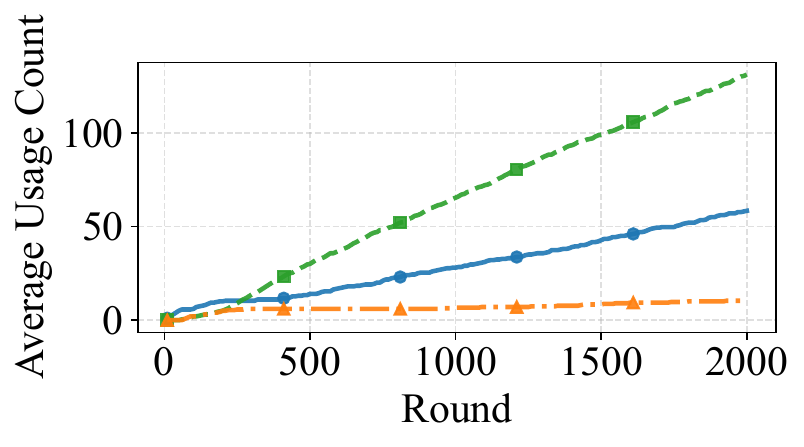}
        \caption{Model prototype = LeNet5}
        \label{fig:mnist_135_5000_use}
    \end{subfigure}
    \caption{Compression strategy selection results of different model prototypes on MNIST dataset ($\alpha = 0.5$, $K=5\times 10^3$).}
    \label{fig:Analysis of the behaviors3}
\end{figure}

\begin{figure}[t]
    \centering
    \begin{minipage}{\linewidth}
        \centering
        \includegraphics[height=0.8cm]{fig/legend.pdf}
    \end{minipage}
    \begin{subfigure}[t]{0.48\linewidth}
        \centering
        \includegraphics[width=\textwidth]{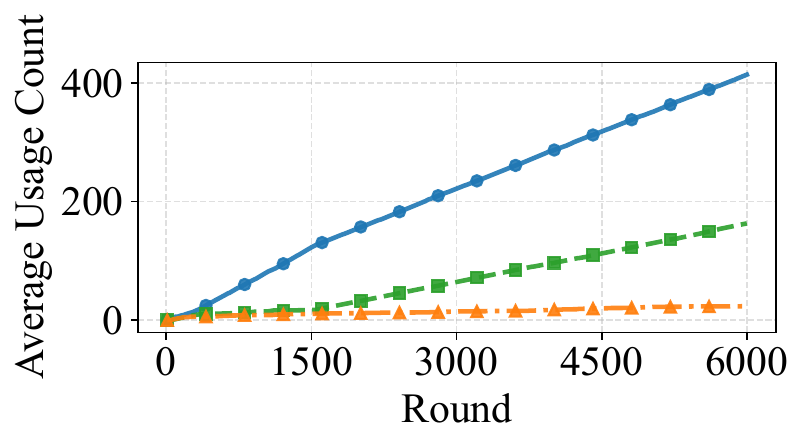}
        \caption{Model prototype = WResNet40-2}
        \label{fig:cifar_135_10000_use}
    \end{subfigure}
    \hfill
    \begin{subfigure}[t]{0.48\linewidth}
        \centering
        \includegraphics[width=\textwidth]{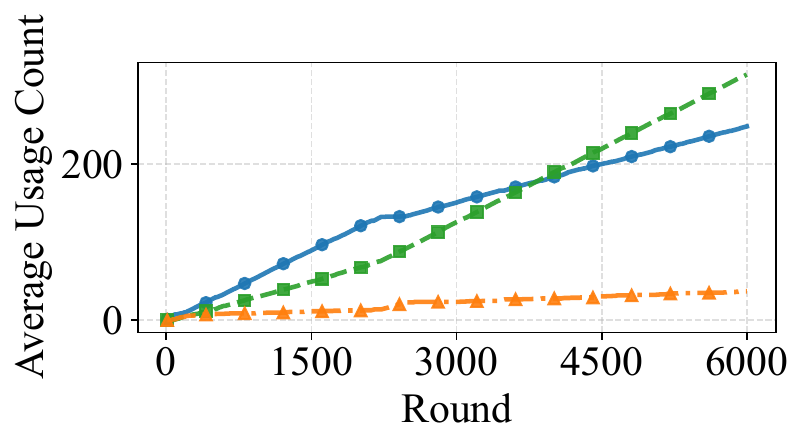}
        \caption{Model prototype = ResNet18}
        \label{fig:cifar_024_10000_use}
    \end{subfigure}
    \caption{Compression strategy selection results of different model prototypes on CIFAR-10 dataset ($\alpha = 0.5$, $K=10^5$).}
    \label{fig:Analysis of the behaviors4}
\end{figure}

\begin{figure}[t]
    \centering
    \begin{minipage}{\linewidth}
        \centering
        \includegraphics[height=0.8cm]{fig/legend.pdf}
    \end{minipage}

    \begin{subfigure}[t]{0.48\linewidth}
        \centering
        \includegraphics[width=\textwidth]{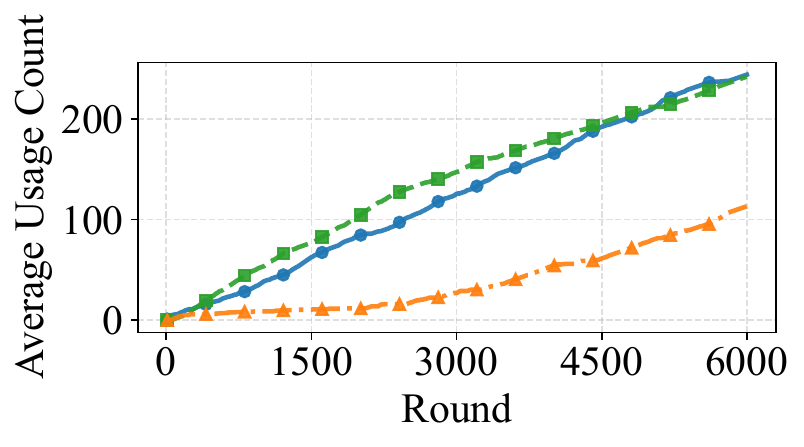}
        \caption{Model prototype = WResNet40-2}
        \label{fig:cifar_135_2000000_use}
    \end{subfigure}
    \hfill
    \begin{subfigure}[t]{0.48\linewidth}
        \centering
        \includegraphics[width=\textwidth]{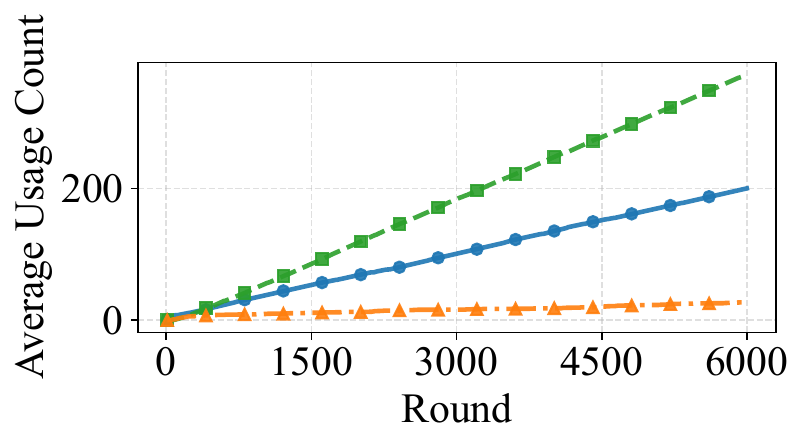}
        \caption{Model prototype = ResNet18}
        \label{fig:cifar_024_2000000_use}
    \end{subfigure}

    \caption{Compression strategy selection results of different model prototypes on CIFAR-10 dataset ($\alpha = 0.5$, $K=2\times 10^6$).}
    \label{fig:Analysis of the behaviors5}
\end{figure}

\subsubsection{Sensitivity Analysis}
As shown in Fig.~\ref{fig:alfa}, the effects of the weighting parameters $\rho$ and $\beta$ in (\ref{eq:deltaL}) are evaluated on different datasets. The parameter setting $\rho=0.6$ and $\beta=0.4$ achieves stable performance across both datasets.

For adaptive strategy selection, the sensitivity of the hyperparameter $c$ in \eqref{eq:maxr} is further investigated. 
As shown in Fig.~\ref{fig:epsilon}, different values of $c$ are evaluated to analyze their impacts on the performance of the proposed method. The results demonstrate that $c=3$ achieves a good balance between exploration and exploitation, providing relatively favorable performance across different datasets.

\begin{figure}[t]
    \centering
    \begin{subfigure}[t]{0.32\linewidth}
        \centering
         \includegraphics[width=\linewidth]{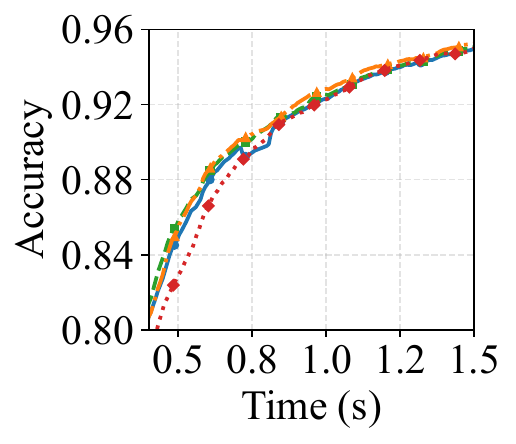}
         \caption{MNIST}
         \label{fig:mnist_alfa}
    \end{subfigure}
    \hfill
    \begin{subfigure}[t]{0.32\linewidth}
        \centering
         \includegraphics[width=\linewidth]{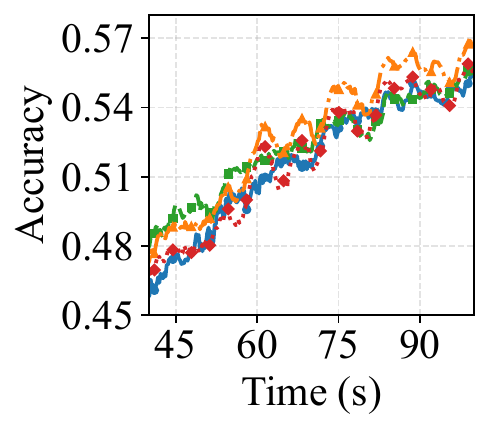}
         \caption{CIFAR-10}
         \label{fig:cifar_alfa}
    \end{subfigure}
    \hfill
    \begin{minipage}[t]{0.24\linewidth}
        \centering
        \includegraphics[width=\linewidth]{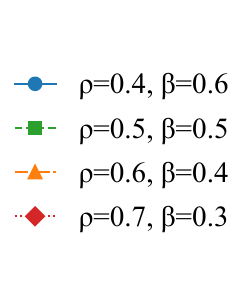}
    \end{minipage}

     \caption{Impact of the weighting parameters $\rho$ and $\beta$ on the performance under different datasets.}
     \label{fig:alfa}
\end{figure}

\begin{figure}[t]
    \centering
    \begin{subfigure}[t]{0.32\linewidth}
        \centering
         \includegraphics[width=\linewidth]{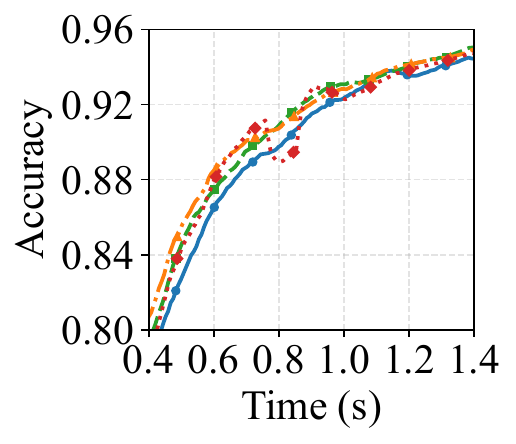}
         \caption{MNIST}
         \label{fig:epsilon_t_mnist}
    \end{subfigure}
    \hfill
    \begin{subfigure}[t]{0.32\linewidth}
        \centering
         \includegraphics[width=\linewidth]{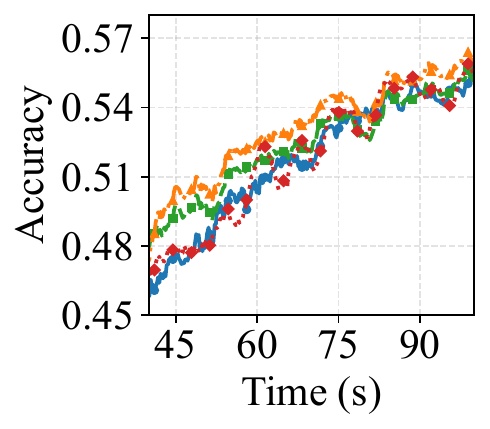}
         \caption{CIFAR-10}
         \label{fig:epsilon_t_cifar}
    \end{subfigure}
    \hfill
    \begin{minipage}[t]{0.24\linewidth}
        \centering
        \includegraphics[width=\linewidth]{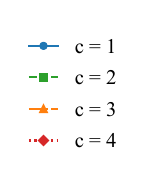}
    \end{minipage}

     \caption{Impact of the hyperparameter $c$ in the probability $\epsilon_t$ on the performance across different datasets.}
     \label{fig:epsilon}
\end{figure}

\section{Conclusion}
This paper proposes ASCEND, an adaptive heterogeneous compression framework for FedKD, where clients with heterogeneous computation and communication resources can dynamically select suitable compression strategies.
The compression strategy selection problem is formulated as a non-stationary stochastic MAB problem.
Each client adaptively chooses among Top-$K$, Random-$K$, and Periodic-$K$ using an EMA-enhanced $\epsilon$-greedy policy with a time-aware reward.
Experimental results on MNIST and CIFAR-10 under non-IID data and different bandwidth capacities demonstrate that ASCEND consistently achieves a favorable balance between model accuracy and training latency.
Future work will investigate the joint optimization of compression strategy selection and compression level adaptation, where the retained gradient budget can be dynamically adjusted according to model performance and system conditions.

\bibliographystyle{IEEEtran}
\bibliography{mybibfilei}

@inproceedings{mcmahan2017communication,
  title={Communication-efficient learning of deep networks from decentralized data},
  author={McMahan, Brendan and Moore, Eider and Ramage, Daniel and Hampson, Seth and y Arcas, Blaise Aguera},
  booktitle={Artificial intelligence and statistics},
  pages={1273--1282},
  year={2017},
  organization={PMLR}
}

@article{li2020federated,
  title={Federated optimization in heterogeneous networks},
  author={Li, Tian and Sahu, Anit Kumar and Zaheer, Manzil and Sanjabi, Maziar and Talwalkar, Ameet and Smith, Virginia},
  journal={Proceedings of Machine learning and systems},
  volume={2},
  pages={429--450},
  year={2020}
}

@article{wen2017terngrad,
  title={Terngrad: Ternary gradients to reduce communication in distributed deep learning},
  author={Wen, Wei and Xu, Cong and Yan, Feng and Wu, Chunpeng and Wang, Yandan and Chen, Yiran and Li, Hai},
  journal={Advances in neural information processing systems},
  volume={30},
  year={2017}
}

@article{alistarh2017qsgd,
  title={QSGD: Communication-efficient SGD via gradient quantization and encoding},
  author={Alistarh, Dan and Grubic, Demjan and Li, Jerry and Tomioka, Ryota and Vojnovic, Milan},
  journal={Advances in neural information processing systems},
  volume={30},
  year={2017}
}

@article{kairouz2021advances,
  title={Advances and open problems in federated learning},
  author={Kairouz, Peter and McMahan, H Brendan and Avent, Brendan and Bellet, Aur{\'e}lien and Bennis, Mehdi and Bhagoji, Arjun Nitin and Bonawitz, Kallista and Charles, Zachary and Cormode, Graham and Cummings, Rachel and others},
  journal={Foundations and trends{\textregistered} in machine learning},
  volume={14},
  number={1--2},
  pages={1--210},
  year={2021},
  publisher={Now Publishers, Inc.}
}

@inproceedings{han2020adaptive,
  title={Adaptive gradient sparsification for efficient federated learning: An online learning approach},
  author={Han, Pengchao and Wang, Shiqiang and Leung, Kin K},
  booktitle={2020 IEEE 40th international conference on distributed computing systems (ICDCS)},
  pages={300--310},
  year={2020},
  organization={IEEE}
}

@article{lin2020ensemble,
  title={Ensemble distillation for robust model fusion in federated learning},
  author={Lin, Tao and Kong, Lingjing and Stich, Sebastian U and Jaggi, Martin},
  journal={Advances in neural information processing systems},
  volume={33},
  pages={2351--2363},
  year={2020}
}

@article{lin2017deep,
  title={Deep gradient compression: Reducing the communication bandwidth for distributed training},
  author={Lin, Yujun and Han, Song and Mao, Huizi and Wang, Yu and Dally, William J},
  journal={arXiv preprint arXiv:1712.01887},
  year={2017}
}

@article{auer2002finite,
  title={Finite-time analysis of the multiarmed bandit problem},
  author={Auer, Peter and Cesa-Bianchi, Nicolo and Fischer, Paul},
  journal={Machine learning},
  volume={47},
  number={2},
  pages={235--256},
  year={2002},
  publisher={Springer}
}

@article{stich2018local,
  title={Local SGD converges fast and communicates little},
  author={Stich, Sebastian U},
  journal={arXiv preprint arXiv:1805.09767},
  year={2018}
}

@article{hinton2015distilling,
  title={Distilling the knowledge in a neural network},
  author={Hinton, Geoffrey and Vinyals, Oriol and Dean, Jeff},
  journal={arXiv preprint arXiv:1503.02531},
  year={2015}
}

@article{aji2017sparse,
  title={Sparse communication for distributed gradient descent},
  author={Aji, Alham Fikri and Heafield, Kenneth},
  journal={arXiv preprint arXiv:1704.05021},
  year={2017}
}

@article{basu2019qsparse,
  title={Qsparse-local-SGD: Distributed SGD with quantization, sparsification and local computations},
  author={Basu, Debraj and Data, Deepesh and Karakus, Can and Diggavi, Suhas},
  journal={Advances in Neural Information Processing Systems},
  volume={32},
  year={2019}
}

@inproceedings{haddadpour2021federated,
  title={Federated learning with compression: Unified analysis and sharp guarantees},
  author={Haddadpour, Farzin and Kamani, Mohammad Mahdi and Mokhtari, Aryan and Mahdavi, Mehrdad},
  booktitle={International Conference on Artificial Intelligence and Statistics},
  pages={2350--2358},
  year={2021},
  organization={PMLR}
}

@article{wu2022communication,
  title={Communication-efficient federated learning via knowledge distillation},
  author={Wu, Chuhan and Wu, Fangzhao and Lyu, Lingjuan and Huang, Yongfeng and Xie, Xing},
  journal={Nature communications},
  volume={13},
  number={1},
  pages={2032},
  year={2022},
  publisher={Nature Publishing Group UK London}
}

@article{li2019fedmd,
  title={Fedmd: Heterogenous federated learning via model distillation},
  author={Li, Daliang and Wang, Junpu},
  journal={arXiv preprint arXiv:1910.03581},
  year={2019}
}

@inproceedings{karimireddy2020scaffold,
  title={Scaffold: Stochastic controlled averaging for federated learning},
  author={Karimireddy, Sai Praneeth and Kale, Satyen and Mohri, Mehryar and Reddi, Sashank and Stich, Sebastian and Suresh, Ananda Theertha},
  booktitle={International conference on machine learning},
  pages={5132--5143},
  year={2020},
  organization={PMLR}
}

@article{lecun2002gradient,
  title={Gradient-based learning applied to document recognition},
  author={LeCun, Yann and Bottou, L{\'e}on and Bengio, Yoshua and Haffner, Patrick},
  journal={Proceedings of the IEEE},
  volume={86},
  number={11},
  pages={2278--2324},
  year={2002},
  publisher={Ieee}
}

@techreport{krizhevsky2009learning,
  title       = {Learning multiple layers of features from tiny images},
  author      = {Krizhevsky, Alex},
  institution = {Department of Computer Science, University of Toronto},
  address     = {Toronto, ON, Canada},
  type        = {Tech. {Rep.}},
  year        = {2009}
}

@article{zhang2021parameter,
  title={Parameterized Knowledge Transfer for Personalized Federated Learning},
  author={Zhang, Jie and Guo, Song and Ma, Xiaosong and Wang, Haozhao and Xu, Wencao and Wu, Feijie},
  journal={arXiv preprint arXiv:2111.02862},
  year={2021}
}

@article{li2019convergence,
  title={On the convergence of fedavg on non-iid data},
  author={Li, Xiang and Huang, Kaixuan and Yang, Wenhao and Wang, Shusen and Zhang, Zhihua},
  journal={arXiv preprint arXiv:1907.02189},
  year={2019}
}

@inproceedings{shi2019convergence,
  title={A convergence analysis of distributed SGD with communication-efficient gradient sparsification.},
  author={Shi, Shaohuai and Zhao, Kaiyong and Wang, Qiang and Tang, Zhenheng and Chu, Xiaowen},
  booktitle={IJCAI},
  pages={3411--3417},
  year={2019}
}

@inproceedings{xu2022detached,
  title={Detached error feedback for distributed SGD with random sparsification},
  author={Xu, An and Huang, Heng},
  booktitle={International conference on machine learning},
  pages={24550--24575},
  year={2022},
  organization={PMLR}
}

@inproceedings{yin2018byzantine,
  title={Byzantine-robust distributed learning: Towards optimal statistical rates},
  author={Yin, Dong and Chen, Yudong and Kannan, Ramchandran and Bartlett, Peter},
  booktitle={International conference on machine learning},
  pages={5650--5659},
  year={2018},
  organization={Pmlr}
}

@ARTICLE{11372971,
  author={Han, Pengchao and Yin, Zhenshuai and Liu, Chang and Fang, Yi},
  journal={IEEE Transactions on Cognitive Communications and Networking}, 
  title={Agentic {AI}-Driven Federated Feature Distillation for Adaptive Resource–Performance Tradeoffs in Wireless Edge Networks}, 
  year={2026},
  volume={12},
  number={},
  pages={6076-6088},
  doi={10.1109/TCCN.2026.3661503}}

@ARTICLE{10606337,
  author={Han, Pengchao and Shi, Xingyan and Huang, Jianwei},
  journal={IEEE Journal on Selected Areas in Communications}, 
  title={FedAL: Black-Box Federated Knowledge Distillation Enabled by Adversarial Learning}, 
  year={2024},
  volume={42},
  number={11},
  pages={3064-3077},
  doi={10.1109/JSAC.2024.3431516}}

@book{sutton2018reinforcement,
  title={Reinforcement Learning: An Introduction},
  author={Sutton, Richard S. and Barto, Andrew G.},
  edition={Second},
  publisher={MIT Press},
  address={Cambridge, MA},
  year={2018}
}

@article{tan2022fedproto,
  title={FedProto: Federated Prototype Learning across Heterogeneous Clients},
  author={Tan, Yue and Long, Guodong and Liu, Lu and Zhou, Tianyi and Jiang, Jing and Zhang, Chengqi},
  journal={Proceedings of the AAAI Conference on Artificial Intelligence},
  volume={36},
  number={8},
  pages={8432--8440},
  year={2022}
}

@article{zhang2023fedtgp,
  title={FedTGP: Federated Learning with Global Prototypes for Heterogeneous Models},
  author={Zhang, Jianqing and Chen, Yang and Wu, Di and Li, Bo},
  journal={IEEE Transactions on Mobile Computing},
  year={2023}
}

@article{han2025rethinking,
  title={Rethinking Knowledge Distillation in Collaborative Machine Learning: Memory, Knowledge, and Their Interactions},
  author={Han, Pengchao and Huang, Xi and Fang, Yi and Han, Guojun},
  journal={IEEE Transactions on Network Science and Engineering},
  volume={12},
  number={6},
  pages={4498--4530},
  year={2025},
  publisher={IEEE},
  doi={10.1109/TNSE.2025.3572362}
}

@inproceedings{chen2018adacomp,
  title={AdaComp: Adaptive Residual Gradient Compression for Data-Parallel Distributed Training},
  author={Chen, Chia-Yu and Choi, Jungwook and Brand, Daniel and Agrawal, Ankur and Zhang, Wei and Gopalakrishnan, Kailash},
  booktitle={Proceedings of the AAAI Conference on Artificial Intelligence},
  volume={32},
  number={1},
  year={2018}
}

@inproceedings{koloskova2020unified,
  title     = {A Unified Theory of Decentralized {SGD} with Changing Topology and Local Updates},
  author    = {Koloskova, Anastasia and Loizou, Nicolas and Boreiri, Sadra and Jaggi, Martin and Stich, Sebastian U.},
  booktitle = {Proceedings of the 37th International Conference on Machine Learning},
  pages     = {5381--5393},
  year      = {2020},
  volume    = {119},
  series    = {Proceedings of Machine Learning Research},
  publisher = {PMLR}
}

@inproceedings{woodworth2020minibatch,
  title     = {Minibatch vs Local {SGD} for Heterogeneous Distributed Learning},
  author    = {Woodworth, Blake E. and Patel, Kumar Kshitij and Srebro, Nathan},
  booktitle = {Advances in Neural Information Processing Systems},
  volume    = {33},
  year      = {2020}
}

@inproceedings{karimireddy2019errorfeedback,
  title     = {Error Feedback Fixes {SignSGD} and Other Gradient Compression Schemes},
  author    = {Karimireddy, Sai Praneeth and Rebjock, Quentin and Stich, Sebastian U. and Jaggi, Martin},
  booktitle = {Proceedings of the 36th International Conference on Machine Learning},
  pages     = {3252--3261},
  year      = {2019},
  volume    = {97},
  series    = {Proceedings of Machine Learning Research},
  publisher = {PMLR}
}

@inproceedings{richtarik2021ef21,
  title     = {{EF21}: A New, Simpler, Theoretically Better, and Practically Faster Error Feedback},
  author    = {Richt{\'a}rik, Peter and Sokolov, Igor and Fatkhullin, Ilyas},
  booktitle = {Advances in Neural Information Processing Systems},
  volume    = {34},
  year      = {2021}
}

@inproceedings{garivier2011switching,
  title     = {On Upper-Confidence Bound Policies for Switching Bandit Problems},
  author    = {Garivier, Aur{\'e}lien and Moulines, Eric},
  booktitle = {Algorithmic Learning Theory},
  pages     = {174--188},
  year      = {2011},
  publisher = {Springer},
  series    = {Lecture Notes in Computer Science},
  volume    = {6925},
  doi       = {10.1007/978-3-642-24412-4_16}
}

@inproceedings{liu2018change,
  title     = {A Change-Detection Based Framework for Piecewise-Stationary Multi-Armed Bandit Problem},
  author    = {Liu, Fang and Lee, Joohyun and Shroff, Ness},
  booktitle = {Proceedings of the AAAI Conference on Artificial Intelligence},
  volume    = {32},
  number    = {1},
  year      = {2018},
  doi       = {10.1609/aaai.v32i1.11746}
}

@inproceedings{zhu2021fedgen,
  title={Federated learning with heterogeneous data: A generative approach},
  author={Zhu, Zhuangdi and Hong, Junyuan and Zhou, Jiayu},
  booktitle={ICML},
  pages={12804--12814},
  year={2021},
  organization={PMLR}
}
\clearpage

\appendix
To facilitate understanding in the theoretical analysis, we first present the overall workflow of a single communication round in Figure~\ref{fig:training_pipeline}.
\begin{figure}[t]
\centering
\includegraphics[width=0.95\linewidth]{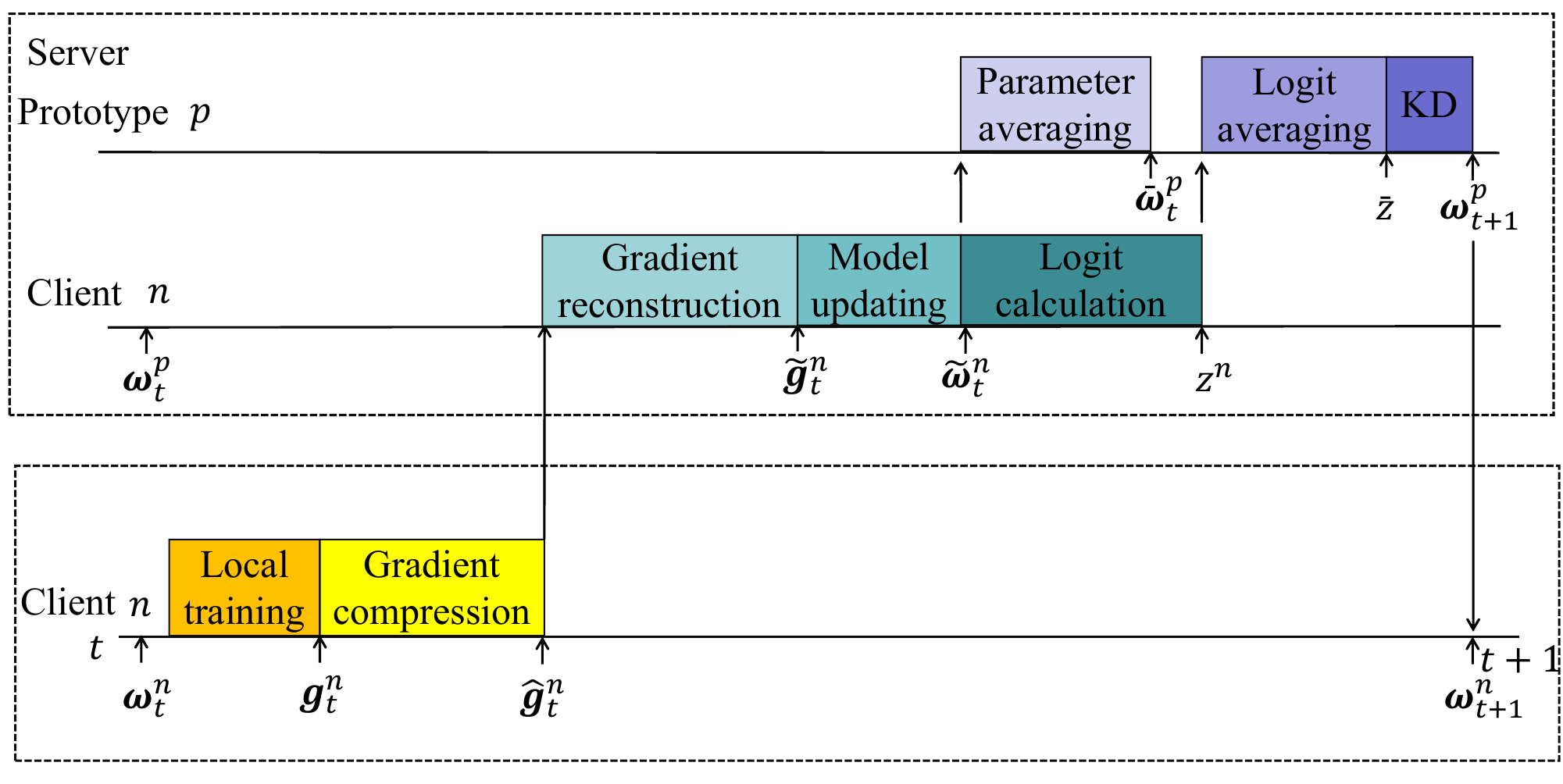}
\caption{The pipeline of a single communication round in FedKD with heterogeneous compression.}
\label{fig:training_pipeline}
\end{figure}

\subsection{Proof of Theorem \ref{thm:nonconvex_convergence} }
We rearrange the training process from the perspective of each prototype. Specifically, at the beginning of training round $t$, the parameter of prototype $p$ is $\boldsymbol{\omega}_t^p$. 
The server broadcasts the updated prototype model after each KD step, all clients belonging to the same prototype start each communication round with identical model parameters, i.e.,
\begin{equation}
\boldsymbol{\omega}_t^n
=
\boldsymbol{\omega}_t^p,
\quad
\forall n\in\mathcal{N}_p .
\label{eq:same_prototype_initial}
\end{equation}
For local training, we define the local optimization objective for prototype $p$ as
\begin{equation}
f_p(\boldsymbol{\omega})
=
\frac{1}{|\mathcal{N}_p|}
\sum_{n\in\mathcal{N}_p}
f_n(\boldsymbol{\omega}),
\label{eq:prototype_objective}
\end{equation}
where $f_p(\boldsymbol{\omega})$ denotes the prototype-wise empirical loss.

Based on the received compressed gradients from client $n$ in round $t$, the server reconstructs the corresponding gradient vector $\tilde{\boldsymbol g}_t^n$.
The reconstructed gradient is used to update the local model of client $n$ by
\begin{equation}
\tilde{\boldsymbol{\omega}}_t^n
=
\boldsymbol{\omega}_t^p
-
\eta_t
\tilde{\boldsymbol g}_t^n.
\label{eq:local_update_theory}
\end{equation}
The server performs parameter-averaging for prototype aggregation as formulated in \eqref{eq:parameter_averaging}.
Substituting \eqref{eq:local_update_theory} into
\eqref{eq:parameter_averaging}, the prototype update can be equivalently written as
\begin{equation}
\bar{\boldsymbol{\omega}}_t^p
=
\boldsymbol{\omega}_t^p
-
\eta_t
\bar{\boldsymbol g}_t^p ,
\label{eq:prototype_gradient_form}
\end{equation}
where the aggregated reconstructed gradient is defined as
\begin{equation}
\bar{\boldsymbol g}_t^p
=
\frac{1}{|\mathcal{N}_p|}
\sum_{n\in\mathcal{N}_p}
\tilde{\boldsymbol g}_t^n .
\label{eq:prototype_gradient_aggregation}
\end{equation}

After prototype aggregation, the server performs one-step KD on the public dataset $\mathcal{D}_{\mathrm{pub}}$.
The KD update is formulated as

\begin{equation}
\boldsymbol{\omega}_{t+1}^{p}
=
\bar{\boldsymbol{\omega}}_t^p
-
\eta_{KD}
\nabla
\mathcal{L}_{KD}^{p}
(\bar{\boldsymbol{\omega}}_t^p),
\label{eq:kd_update}
\end{equation}
where $\eta_{KD}$ denotes the KD learning rate.

For theoretical analysis, we define the joint optimization objective of prototype $p$ as
\begin{equation}
F_p(\boldsymbol{\omega}^p_t)
=
f_p(\boldsymbol{\omega}^p_t)
+
\lambda
\mathcal{L}_{KD}^{p}
(\boldsymbol{\omega}^p_t),
\label{eq:joint_objective}
\end{equation}
where $\lambda$ balances the local optimization objective and the global knowledge alignment objective.

Define $\eta_{KD}\triangleq \lambda\eta_t$
and substituting \eqref{eq:prototype_gradient_form}  into \eqref{eq:kd_update} yields
\begin{equation}
\boldsymbol{\omega}_{t+1}^{p}
=
\boldsymbol{\omega}_t^p
-
\eta_t
\left(
\bar{\boldsymbol g}_t^p
+
\lambda
\nabla
\mathcal{L}_{KD}^{p}
(\bar{\boldsymbol{\omega}}_t^p)
\right).
\label{eq:combined_update}
\end{equation}

For convenience, define the aggregated residual and stochastic gradient noise
for prototype $p$ as
\begin{equation}
\boldsymbol e_t^p
=
\frac{1}{|\mathcal N_p|}
\sum_{n\in\mathcal N_p}
\boldsymbol e_t^n ,
\end{equation}
and
\begin{equation}
\boldsymbol\xi_t^p
=
\frac{1}{|\mathcal N_p|}
\sum_{n\in\mathcal N_p}
(
\boldsymbol g_t^n
-
\nabla f_n(\boldsymbol\omega_t^p)
).
\end{equation}
Applying Assumption~\ref{asm:Unbiased Stochastic Gradient} and~\ref{asm:Bounded Stochastic Gradient Variance},
the stochastic gradient noise satisfies
\begin{equation}
\mathbb E
[
\boldsymbol\xi_t^p
|
\boldsymbol\omega_t^p
]
=
0,
\end{equation}
and
\begin{equation}
\mathbb E[
\|
\boldsymbol\xi_t^p
\|^2]
\leq
\sigma^2 .
\end{equation}

We first propose the following lemma to show the smoothness of $F_p(\boldsymbol{\omega})$.

\begin{lem}[Smoothness of the Joint Objective]
\label{lem:Smoothness-fp}
The joint objective function $F_p(\boldsymbol{\omega})$
defined in \eqref{eq:joint_objective} is $L$-smooth. Specifically, there exists a constant
$L=L_l+\lambda L_p>0$ such that
\begin{equation}
\left\|
\nabla F_p(\boldsymbol{\omega})
-
\nabla F_p(\boldsymbol{\omega}')
\right\|
\leq
L
\left\|
\boldsymbol{\omega}
-
\boldsymbol{\omega}'
\right\|,
\quad
\forall
\boldsymbol{\omega},\boldsymbol{\omega}' .
\label{eq:smoothnessF_p}
\end{equation}
\end{lem}

\begin{proof}
According to the definition of the joint optimization objective in
\eqref{eq:joint_objective}, therefore, for any two model parameters
$\boldsymbol{\omega}$ and $\boldsymbol{\omega}'$, the gradient difference can be written as
\begin{equation}
\begin{aligned}
&
\nabla F_p(\boldsymbol{\omega})
-
\nabla F_p(\boldsymbol{\omega}')
\\
&=
\nabla f_p(\boldsymbol{\omega})
-
\nabla f_p(\boldsymbol{\omega}')
+
\lambda
\left(
\nabla \mathcal{L}_{KD}^{p}(\boldsymbol{\omega})
-
\nabla \mathcal{L}_{KD}^{p}(\boldsymbol{\omega}')
\right).
\end{aligned}
\end{equation}
Taking the norm on both sides and applying the triangle inequality, we have
\begin{equation}
\begin{aligned}
&
\left\|
\nabla F_p(\boldsymbol{\omega})
-
\nabla F_p(\boldsymbol{\omega}')
\right\|
\\
&\leq
\left\|
\nabla f_p(\boldsymbol{\omega})
-
\nabla f_p(\boldsymbol{\omega}')
\right\|
\\
&\quad
+
\lambda
\left\|
\nabla \mathcal{L}_{KD}^{p}(\boldsymbol{\omega})
-
\nabla \mathcal{L}_{KD}^{p}(\boldsymbol{\omega}')
\right\|.
\end{aligned}
\end{equation}
Applying Assumption~\ref{asm:Smoothness-l} and Assumption~\ref{asm:Smoothness-p}  into the gradient difference bound yields
\begin{equation}
\begin{aligned}
&
\left\|
\nabla F_p(\boldsymbol{\omega})
-
\nabla F_p(\boldsymbol{\omega}')
\right\|
\\
&\leq
L_l
\left\|
\boldsymbol{\omega}
-
\boldsymbol{\omega}'
\right\|
+
\lambda L_p
\left\|
\boldsymbol{\omega}
-
\boldsymbol{\omega}'
\right\|
\\
&=
(L_l+\lambda L_p)
\left\|
\boldsymbol{\omega}
-
\boldsymbol{\omega}'
\right\|.
\end{aligned}
\end{equation}
By defining the smoothness constant as
\begin{equation}
\label{eq:L:}
L=L_l+\lambda L_p,
\end{equation}
we obtain
\begin{equation}
\left\|
\nabla F_p(\boldsymbol{\omega})
-
\nabla F_p(\boldsymbol{\omega}')
\right\|
\leq
L
\left\|
\boldsymbol{\omega}
-
\boldsymbol{\omega}'
\right\|,
\end{equation}
which completes the proof.
\end{proof}

We then propose the following lemma to capture the deviation in each round $t$.
\begin{lem}[Error-feedback Decomposition under Heterogeneous Compression]
\label{lem:gradient_deviation}
For each prototype $p$, 
the deviation between the aggregated reconstructed gradient $\bar{\boldsymbol g}_t^p$ and the true prototype gradient satisfies
\begin{equation}
\bar{\boldsymbol g}_t^p
-
\nabla f_p(\boldsymbol{\omega}_t^p)
=
\boldsymbol e_{t-1}^{p}
-
\boldsymbol e_t^{p}
+
\boldsymbol{\xi}_t^{p}.
\label{eq:ef_gradient_decomposition}
\end{equation}
\end{lem}

\begin{prf}
According to the error-feedback update,
the residual vector at client $n$ satisfies
\begin{equation}
\boldsymbol e_t^n
=
\boldsymbol u_t^n
-
\tilde{\boldsymbol g}_t^n.
\label{eq:ef_residual_definition}
\end{equation}
Using \eqref{eq:residual_addition_III},
therefore,
\begin{equation}
\tilde{\boldsymbol g}_t^n
-
\boldsymbol g_t^n
=
\boldsymbol e_{t-1}^n
-
\boldsymbol e_t^n .
\label{eq:ef_telescoping_identity}
\end{equation}
By adding and subtracting
$\nabla f_n(\boldsymbol{\omega}_t^p)$,
we obtain
\begin{align}
\tilde{\boldsymbol g}_t^n
-
\nabla f_n(\boldsymbol{\omega}_t^p)
=&
\boldsymbol e_{t-1}^n
-
\boldsymbol e_t^n
\nonumber\\
&
+
\left(
\boldsymbol g_t^n
-
\nabla f_n(\boldsymbol{\omega}_t^p)
\right).
\label{eq:local_ef_decomposition}
\end{align}
Taking the average over all clients belonging to prototype $p$ gives
\begin{align}
&
\bar{\boldsymbol g}_t^p
-
\nabla f_p(\boldsymbol{\omega}_t^p)
\nonumber\\
=&
\boldsymbol e_{t-1}^{p}
-
\boldsymbol e_t^{p}
+
\boldsymbol{\xi}_t^p ,
\end{align}
which completes the proof.
\end{prf}

Now, we prove Theorem \ref{thm:nonconvex_convergence}.
\begin{prf}
According to Lemma \ref{lem:Smoothness-fp}, the joint objective function
$F_p(\boldsymbol{\omega})$
is $L$-smooth. Therefore, 
we obtain
\begin{align}
F_p(\boldsymbol{\omega}_{t+1}^{p})
\leq&
F_p(\boldsymbol{\omega}_{t}^{p})
+
\left\langle
\nabla F_p(\boldsymbol{\omega}_{t}^{p}),
\boldsymbol{\omega}_{t+1}^{p}
-
\boldsymbol{\omega}_{t}^{p}
\right\rangle
\nonumber\\
&
+
\frac{L}{2}
\|
\boldsymbol{\omega}_{t+1}^{p}
-
\boldsymbol{\omega}_{t}^{p}
\|^2 .
\label{eq:smoothness_tound}
\end{align}

According to the prototype update rule after parameter aggregation and KD, the prototype model is updated following \eqref{eq:combined_update}.

The definition of the aggregated gradient deviation,
\begin{equation}
\boldsymbol{\delta}_{t}^{p}
=
\bar{\boldsymbol g}_{t}^{p}
-
\nabla f_p(\boldsymbol{\omega}_{t}^{p}).
\label{eq:delta_definition_theorem}
\end{equation}
According to Lemma~\ref{lem:gradient_deviation},
the aggregated gradient deviation can be written as
\begin{equation}
\boldsymbol{\delta}_t^p
=
\boldsymbol e_{t-1}^{p}
-
\boldsymbol e_t^{p}
+
\boldsymbol{\xi}_t^p .
\label{eq:new_delta_definition}
\end{equation}

From \eqref{eq:joint_objective} and \eqref{eq:delta_definition_theorem}, the update direction can be rewritten as
\begin{equation}
\bar{\boldsymbol g}_{t}^{p}
+
\lambda
\nabla
\mathcal L_{KD}^{p}
(\bar{\boldsymbol{\omega}}_{t}^{p})
=
\nabla F_p(\boldsymbol{\omega}_{t}^{p})
+
\boldsymbol{\delta}_{t}^{p}.
\label{eq:update_direction_decomposition}
\end{equation}
Therefore, by combining the result with \eqref{eq:combined_update}, we have
\begin{equation}
\boldsymbol{\omega}_{t+1}^{p}
-
\boldsymbol{\omega}_{t}^{p}
=
-\eta_t
\left(
\nabla F_p(\boldsymbol{\omega}_{t}^{p})
+
\boldsymbol{\delta}_{t}^{p}
\right).
\label{eq:update_difference}
\end{equation}
Substituting
\eqref{eq:update_difference}
into
\eqref{eq:smoothness_tound}, we have
\begin{align}
F_p(\boldsymbol{\omega}_{t+1}^{p})
\leq&
F_p(\boldsymbol{\omega}_{t}^{p})
-
{\eta_t}
\left\langle
\nabla F_p(\boldsymbol{\omega}_{t}^{p}),
\nabla F_p(\boldsymbol{\omega}_{t}^{p})
+
\boldsymbol{\delta}_{t}^{p}
\right\rangle
\nonumber\\
&
+
\frac{L\eta_t^2}{2}
\|
\nabla F_p(\boldsymbol{\omega}_{t}^{p})
+
\boldsymbol{\delta}_{t}^{p}
\|^2 .
\label{eq:substitution_update}
\end{align}

For the inner-product term of \eqref{eq:substitution_update},
we have
\begin{align}
&
-\eta_t
\left\langle
\nabla F_p(\boldsymbol{\omega}_{t}^{p}),
\nabla F_p(\boldsymbol{\omega}_{t}^{p})+\boldsymbol{\delta}_t^p
\right\rangle
\nonumber\\
=&
-\eta_t
\|
\nabla F_p(\boldsymbol{\omega}_{t}^{p})
\|^2
-
\eta_t
\left\langle
\nabla F_p(\boldsymbol{\omega}_{t}^{p}),
\boldsymbol e_{t-1}^{p}
-
\boldsymbol e_t^{p}
\right\rangle
\nonumber\\
&
-
\eta_t
\left\langle
\nabla F_p(\boldsymbol{\omega}_{t}^{p}),
\boldsymbol{\xi}_t^{p}
\right\rangle .
\end{align}
Taking expectation and according to Assumption~\ref{asm:Unbiased Stochastic Gradient},
the stochastic gradient noise term vanishes:
\begin{equation}
\mathbb E\left[
\left\langle
\nabla F_p(\boldsymbol{\omega}_{t}^{p}),
\boldsymbol{\xi}_t^p
\right\rangle\right]
=
0 .
\end{equation}
Therefore,
\begin{align}
&
\mathbb E
\left[
-\eta_t
\left\langle
\nabla F_p(\boldsymbol{\omega}_{t}^{p}),
\nabla F_p+\boldsymbol{\delta}_t^p
\right\rangle
\right]
\nonumber\\
=&
-\eta_t
\mathbb E
\left[
\|
\nabla F_p(\boldsymbol{\omega}_{t}^{p})
\|^2
\right]
-
\eta_t
\mathbb E
\left[
\left\langle
\nabla F_p(\boldsymbol{\omega}_{t}^{p}),
\boldsymbol e_{t-1}^{p}
-
\boldsymbol e_t^{p}
\right\rangle 
\right].
\label{eq:ef_inner_product}
\end{align}

For the second-order term of \eqref{eq:substitution_update},
according to
\[
\|a+b\|^2
\leq
2\|a\|^2+2\|b\|^2 ,
\]
we have
\begin{align}
&
\frac{L\eta_t^2}{2}
\|
\nabla F_p(\boldsymbol{\omega}_{t}^{p})+\boldsymbol{\delta}_t^p
\|^2
\leq
L\eta_t^2
\|\nabla F_p(\boldsymbol{\omega}_{t}^{p})\|^2
+
L\eta_t^2
\|\boldsymbol{\delta}_t^p\|^2 .
\label{eq:second_order_ef}
\end{align}
According to Lemma~\ref{lem:gradient_deviation},
the aggregated gradient deviation can be decomposed as
\eqref{eq:new_delta_definition}.
Using the inequality
\[
\|a+b+c\|^2
\leq
3\|a\|^2
+
3\|b\|^2
+
3\|c\|^2 ,
\]
we obtain
\begin{align}
\mathbb E
\left[ 
\|
\boldsymbol{\delta}_t^p
\|^2
\right]
\leq&
3
\mathbb E
\left[ 
\|
\boldsymbol e_{t-1}^{p}
\|^2
\right]
+
3
\mathbb E
\left[ 
\|
\boldsymbol e_t^{p}
\|^2
\right]
\nonumber\\
&
+
3
\mathbb E
\left[ 
\|
\boldsymbol\xi_t^{p}
\|^2\right] .
\label{eq:delta_second_moment}
\end{align}
By Assumption~\ref{asm:Bounded Stochastic Gradient Variance} and Assumption~\ref{asm:Bounded Error-Feedback Residual},
we have,
\begin{equation}
\mathbb E
\left[ 
\|
\boldsymbol{\delta}_t^p
\|^2
\right]
\leq
3(2\Gamma_{\max}+\sigma^2)
\triangleq
C_g .
\label{eq:delta_bound_ef}
\end{equation}
where $C_g$ denotes the bounded gradient deviation constant induced by heterogeneous compression and stochastic optimization.

Combining
\eqref{eq:ef_inner_product}, \eqref{eq:second_order_ef} and \eqref{eq:delta_bound_ef}, \eqref{eq:substitution_update} becomes
\begin{align}
&
\mathbb E
\left[ 
F_p(\boldsymbol{\omega}_{t+1}^{p})
 \right]
\nonumber\\
\leq&
\mathbb E
\left[
F_p(\boldsymbol{\omega}_{t}^{p})
\right]
-
(\eta_t-L\eta_t^2)
\mathbb E
\left[ 
\|\nabla F_p(\boldsymbol{\omega}_{t}^{p})\|^2
\right]
\nonumber\\
&
-
\eta_t
\mathbb E
\left[  
\left\langle
\nabla F_p(\boldsymbol{\omega}_{t}^{p}),
\boldsymbol e_{t-1}^{p}
-
\boldsymbol e_t^{p}
\right\rangle
 \right]
\nonumber\\
&
+
L\eta_t^2C_g .
\label{eq:new_descent}
\end{align}
The accumulated residual term can be bounded by the error-feedback telescoping property.
Specifically,
\begin{align}
\sum_{t=0}^{T-1}
(
\boldsymbol e_{t-1}^{p}
-
\boldsymbol e_t^{p}
)
=
\boldsymbol e_{-1}^{p}
-
\boldsymbol e_{T-1}^{p}.
\label{eq:residual_telescoping}
\end{align}

According to the standard error-feedback convergence property
\cite{karimireddy2019errorfeedback} and Assumption~\ref{asm:Bounded Error-Feedback Residual},
the residual accumulation term is bounded by a finite constant:
\begin{equation}
\left|
\sum_{t=0}^{T-1}
\eta_t
\mathbb E
\left[   
\langle
\nabla F_p(\boldsymbol{\omega}_{t}^{p}),
\boldsymbol e_{t-1}^{p}
-
\boldsymbol e_t^{p}
\rangle
\right]
\right|
\leq C_k ,
\end{equation}
where $C_k$ is a finite constant related to the error-feedback residual bound.

Summing
\eqref{eq:new_descent}
from $t=0$ to $T-1$ gives
\begin{align}
&
\sum_{t=0}^{T-1}
(\eta_t-L\eta_t^2)
\mathbb E
\left[
\|\nabla F_p(\boldsymbol{\omega}_t^p)\|^2
\right]
\nonumber\\
\leq&
F_p(\boldsymbol{\omega}_0^p)
-
F_p(\boldsymbol{\omega}_{T}^{p})
+
C_k
+
LT\eta_t^2C_g .
\end{align}
Since
$\eta_t\leq\frac1{4L}$,
\begin{equation}
{\eta_t}-L{\eta_t}^2
\geq
\frac{{3\eta_t}}{4},
\end{equation}
therefore,
\begin{align}
&
\frac{3\eta_t}{4}
\sum_{t=0}^{T-1}
\mathbb E
\left[ 
\|
\nabla F_p(\boldsymbol{\omega}_t^p)
\|^2
\right]
\nonumber\\
\leq&
F_p(\boldsymbol{\omega}_0^p)
-
F_p^\ast
+
C_k
+
LT\eta_t^2C_g .
\end{align}
Dividing both sides by
$\frac{{3\eta_t}T}{4}$,
we obtain
\begin{equation}
\frac1T
\sum_{t=0}^{T-1}
\mathbb E
\left[ 
\|
\nabla F_p(\boldsymbol{\omega}_t^p)
\|^2
\right]
\leq
\frac{
4
(
F_p(\boldsymbol{\omega}_0^p)-F_p^\ast+C_k
)
}
{3\eta_tT}
+
\frac{4L}{3}
\eta_t C_g .
\label{eq:final_convergence_bound}
\end{equation}
Therefore, by choosing
\[
\eta_t=O(1/\sqrt T),
\]
the first term and the second term both decrease as
$O(1/\sqrt T)$.
Therefore,
\begin{equation}
\frac1T
\sum_{t=0}^{T-1}
\mathbb E
\left[ 
\|
\nabla F_p(\boldsymbol{\omega}_t^p)
\|^2
\right]
=
O(1/\sqrt T).
\end{equation}

 \end{prf}

\subsection{Proof of Theorem \ref{thm:regret} }

\subsubsection{Proof of Lemma \ref{lem:ema_tracking}}

\begin{prf}
When strategy $s$ is selected, the EMA utility update follows \eqref{eq:q_update}.

Within the segment $\mathcal{T}_i$, by Assumption~\ref{asm:piecewise_stationary} and \eqref{eq:expected_reward}, the expected reward remains constant:
\begin{equation}
\mathbb{E}[\mathcal{G}_t^n(s)]
=
\mu_i^n(s).
\end{equation}
Subtracting $\mu_i^n(s)$ from both sides of \eqref{eq:q_update} yields
\begin{align}
Q_t^n(s)-\mu_i^n(s)
=&
(1-\eta_q)
(Q_{t-1}^n(s)-\mu_i^n(s))
\nonumber\\
&
+
\eta_q
(\mathcal{G}_t^n(s)-\mu_i^n(s)).
\end{align}
Taking expectation gives
\begin{align}
&
\mathbb{E}
[
Q_t^n(s)-\mu_i^n(s)
]
\nonumber\\
=&
(1-\eta_q)
\mathbb{E}
[
Q_{t-1}^n(s)-\mu_i^n(s)
].
\end{align}
Applying the above recursion for $m$ updates gives
\begin{equation}
\mathbb{E}
[
Q_t^n(s)-\mu_i^n(s)
]
=
(1-\eta_q)^m
\mathbb{E}
[
Q_{r_i}^n(s)-\mu_i^n(s)
].
\end{equation}
Taking the absolute value on both sides gives \eqref{eq:ema_tracking_error}.
\end{prf}

\subsubsection{Proof of Lemma \ref{lem:exploration_regret}}

\begin{prf}
During the exploration stage, ASCEND selects a compression strategy randomly with probability $\epsilon_t$, as defined in \eqref{eq:maxr}. Therefore, the expected regret caused by exploration at communication round $t$ is bounded by
\begin{equation}
\epsilon_t
\left(
\mu_t^n(s_t^{n,*})
-
\mu_t^n(s_t^n)
\right).
\end{equation}

Due to the reward clipping operation, the instantaneous regret is
bounded by
\begin{equation}
\mu_t^n(s_t^{n,*})
-
\mu_t^n(s_t^n)
\leq
\Delta_{\max}.
\end{equation}
Therefore, the cumulative exploration regret after $T$ communication
rounds satisfies
\begin{align}
\mathcal{R}_{\mathrm{explore}}^n(T)
&\leq
\Delta_{\max}
\sum_{t=1}^{T}
\epsilon_t
\nonumber\\
&\leq
c\Delta_{\max}
\sum_{t=1}^{T}
\frac1{\sqrt t}.
\end{align}
Using the integral inequality,
\begin{equation}
\sum_{t=1}^{T}
\frac1{\sqrt t}
\leq
1+
\int_1^T x^{-1/2}dx
\leq
2\sqrt T,
\end{equation}
we obtain \eqref{eq:exploration_regret_bound}.
\end{prf}

\subsubsection{Proof of Lemma \ref{lem:exploitation_regret}}

\begin{prf}
During the exploitation phase, ASCEND selects the compression strategy
according to the current utility estimation:
\begin{equation}
s_t^n
=
\arg\max_{s\in\mathcal S}
Q_t^n(s).
\end{equation}
Define the utility estimation error as
\begin{equation}
v_t^n(s)
=
Q_t^n(s)-\mu_t^n(s).
\label{eq:utility_estimation_error}
\end{equation}
Since $s_t^n$ maximizes the estimated utility,
\begin{equation}
Q_t^n(s_t^n)
\geq
Q_t^n(s_t^{n,*}).
\end{equation}
Substituting
$Q_t^n(s)=\mu_t^n(s)+v_t^n(s)$ gives
\begin{align}
&
\mu_t^n(s_t^n)
+
v_t^n(s_t^n)
\geq
\mu_t^n(s_t^{n,*})
+
v_t^n(s_t^{n,*}).
\end{align}
Therefore,
\begin{align}
&
\mu_t^n(s_t^{n,*})
-
\mu_t^n(s_t^n)
\leq
v_t^n(s_t^n)
-
v_t^n(s_t^{n,*}).
\end{align}
The instantaneous exploitation regret is
\begin{align}
\mu_t^n(s_t^{n,*})
-
\mu_t^n(s_t^n)
&\leq
|v_t^n(s_t^n)|
+
|v_t^n(s_t^{n,*})|
\nonumber\\
&\leq
2
\max_{s\in\mathcal S}
|v_t^n(s)|.
\label{eq:instant_exploit_regret}
\end{align}

Next, we analyze the cumulative estimation error within a stationary segment $\mathcal T_i$. According to Assumption~\ref{asm:piecewise_stationary}, the expected reward remains unchanged within each segment.
The utility estimation error can be decomposed into the EMA bias and the stochastic fluctuation:
\begin{equation}
v_t^n(s)
=
\underbrace{
Q_t^n(s)-\mathbb{E}[Q_t^n(s)]
}_{\text{stochastic fluctuation}}
+
\underbrace{
\mathbb{E}[Q_t^n(s)]-\mu_i^n(s)
}_{\text{EMA bias}} .
\label{eq:error_decomposition}
\end{equation}

According to Assumption~\ref{asm:reward_concentration} and
Lemma~\ref{lem:ema_tracking}, after $H_{i,t}^n(s)$ updates of strategy
$s$, its utility estimation error satisfies
\begin{equation}
\mathbb{E}
\left[
|v_t^n(s)|
\right]
\leq
\frac{C_1}
{\sqrt{H_{i,t}^n(s)\vee 1}}
+
C_2(1-\eta_q)^{H_{i,t}^n(s)},
\label{eq:73}
\end{equation}
where $a\vee b=\max\{a,b\}$, and $C_1$ and $C_2$ are constants determined by the reward variance, EMA coefficient, and initialization error.

Combining \eqref{eq:73} with
\eqref{eq:instant_exploit_regret}, the expected instantaneous
exploitation regret at round $t\in\mathcal T_i$ satisfies
\begin{align}
&\mathbb{E}\!\left[
\mu_t^n(s_t^{n,*})
-
\mu_t^n(s_t^n)
\right]
\nonumber\\
&\leq
C_1
\left[
\frac{1}
{\sqrt{H_{i,t}^n(s_t^n)\vee1}}
+
\frac{1}
{\sqrt{H_{i,t}^n(s_t^{n,*})\vee1}}
\right]
\nonumber\\
&\quad+
C_2
\left[
(1-\eta_q)^{H_{i,t}^n(s_t^n)}
+
(1-\eta_q)^{H_{i,t}^n(s_t^{n,*})}
\right].
\label{eq:instant_exploit_visit}
\end{align}
Under Assumption~\ref{asm:sufficient_sampling}, for every
$s\in\mathcal S$, \eqref{eq:instant_exploit_visit} yields
\begin{align}
&\mathbb{E}\!\left[
\mu_t^n(s_t^{n,*})
-
\mu_t^n(s_t^n)
\right]
\nonumber\\
&\leq
\frac{2C_1}
{\sqrt{\pi_{\min}(t-r_i+1)}}
+
2C_2
(1-\eta_q)^{\pi_{\min}(t-r_i+1)}.
\label{eq:instant_exploit_round}
\end{align}
Summing \eqref{eq:instant_exploit_round} over
$t\in\mathcal T_i$, we obtain
\begin{align}
\mathcal R_{\mathrm{exploit},i}^n
\nonumber
&\leq
\frac{4C_1}{\sqrt{\pi_{\min}}}
\sqrt{|\mathcal T_i|}
+
\frac{2C_2}
{1-(1-\eta_q)^{\pi_{\min}}}.
\label{eq:segment_exploit_regret}
\end{align}
Let
\begin{equation}
C_e
=
\frac{4C_1}{\sqrt{\pi_{\min}}}
+
\frac{2C_2}
{1-(1-\eta_q)^{\pi_{\min}}},
\end{equation}
where $C_e$ is independent of $|\mathcal T_i|$.
Since $|\mathcal T_i|\geq1$, we have
\begin{equation}
\mathcal R_{\mathrm{exploit},i}^n
\leq
C_e\sqrt{|\mathcal T_i|}.
\end{equation}
Finally, summing over all stationary segments gives
\begin{align}
\mathcal R_{\mathrm{exploit}}^n(T)
&=
\sum_{i=1}^{M}
\mathcal R_{\mathrm{exploit},i}^n
\nonumber\\
&\leq
C_e
\sum_{i=1}^{M}
\sqrt{|\mathcal T_i|}
\nonumber\\
&\leq
C_e
\sqrt{
M\sum_{i=1}^{M}|\mathcal T_i|
}
\nonumber\\
&=
C_e\sqrt{MT},
\end{align}
where the second inequality follows from the Cauchy--Schwarz
inequality.
\end{prf}

\subsubsection{Proof of Theorem \ref{thm:regret}}
\begin{prf}
We first analyze the regret incurred during the warm-up phase.

During warm-up, each compression strategy is selected exactly $h$ times to initialize the utility estimation. Therefore, the number of warm-up rounds is $h|\mathcal S|$.

Since the reward is clipped into $\left[\mathcal{G}_{\min},\mathcal{G}_{\max}\right]$, the maximum instantaneous regret is bounded by \eqref{eq:Gmax,min}.
Therefore, the warm-up regret is bounded by
\begin{equation}
\mathcal{R}_{\mathrm{warm}}^n
\leq
h|\mathcal S|\Delta_{\max}.
\label{eq:warmup_regret_bound}
\end{equation}

Combining
\eqref{eq:warmup_regret_bound}, Lemmas \ref{lem:exploration_regret} and \ref{lem:exploitation_regret},
we obtain
\begin{align}
\mathcal{R}_{\mathrm{ASCEND}}^n(T)&
=
\mathcal{R}_{\mathrm{warm}}^n
+
\mathcal{R}_{\mathrm{explore}}^n(T)
+
\mathcal{R}_{\mathrm{exploit}}^n(T),\\
&\leq
h|\mathcal S|\Delta_{\max}
+
2c\Delta_{\max}\sqrt T 
+
C_e\sqrt{MT}.
\end{align}
Therefore,
\begin{equation}
\mathcal{R}_{\mathrm{ASCEND}}^n(T)
=
O(\sqrt{MT}).
\end{equation}
\end{prf}



\vfill

\end{document}